\documentclass[smallextended]{svjour3}
\usepackage{algorithm, algorithmic}
\usepackage{amsmath}
\usepackage{amsfonts}
\usepackage{url}

\smartqed  
\usepackage{graphicx}

\newcommand{\RR}{\mathbb{R}}

\newcommand{\NN}{\mathbb{N}}
\newcommand{\yy}{\mathbf{y}}
\newcommand{\xx}{\mathbf{x}}
\newcommand{\zz}{\mathbf{z}}

\newcommand{\SSS}{\mathrm{S}}
\newcommand{\MM}{\mathcal{M}}
\newcommand{\TT}{\mathcal{T}}
\newcommand{\II}{\mathcal{I}}

\newcommand{\sdp}{{\rm SDP}}

\newcommand{\oomit}[1]{}

\begin{document}
\title{Nonlinear Craig Interpolant Generation
}


\author{Ting Gan         \and
        Bican Xia      \and
        Bai Xue    \and
        Naijun Zhan   \and
        Liyun Dai
}


\institute{T. Gan \at
              School of Computer Science,
		Wuhan University,
		Wuhan, China \\
       \email{ganting@whu.edu.cn}
           \and
           B. Xia \at
              LMAM \& School of Mathematical Sciences, Peking University, Beijing, China\\
  \email{xbc@math.pku.edu.cn}
           \and
           B. Xue \at
           State Key Lab. of Computer Science, Institute of Software, Chinese Academy of Sciences, Beijing, China\\
  \email{xuebai@ios.ac.cn}
           \and
           N. Zhan \at
           State Key Lab. of Computer Science, Institute of Software, Chinese Academy of Sciences\& University of CAS, Beijing, China\\
  \email{znj@ios.ac.cn}
           \and
           L. Dai \at
           RISE \& School of Computer and Information Science, Southwest University, Chongqing, China\\
  \email{dailiyun@swu.edu.cn}
}

\date{Received: date / Accepted: date}

\maketitle

\begin{abstract}
  Interpolation-based techniques have become popular in recent years because of their inherently modular and local reasoning,
	which can scale up existing formal verification techniques like theorem proving, model-checking, abstraction interpretation, and
	so on, while the scalability is the bottleneck of these techniques.
	Craig interpolant generation plays a central role in interpolation-based techniques, and therefore
	has drawn increasing attention. In the literature, there are various works
	done on how to automatically synthesize interpolants for decidable fragments of first-order logic, linear arithmetic, array logic, equality logic with uninterpreted functions (EUF), etc., and their combinations. But Craig interpolant generation for non-linear theory and its combination with the aforementioned theories are
	still in infancy, although some attempts have been done.
	In this paper,
	we first prove that a polynomial interpolant of the form $h(\xx)>0$ exists for two
	mutually contradictory polynomial formulas
	$\phi(\xx,\yy)$ and $\psi(\xx,\zz)$, with the form
	$f_1\ge0\wedge\cdots\wedge f_n\ge0$, where $f_i$ are polynomials in $\xx,\yy$ or $\xx,\zz$, and the quadratic module generated by $f_i$ is Archimedean. Then, we show that synthesizing such interpolant can be reduced to solving a semi-definite programming  problem ($\sdp$).
	In addition, we propose a verification approach to assure the validity of the synthesized interpolant and consequently avoid the unsoundness caused by numerical error in $\sdp$ solving. Then, we discuss how to generalize our approach to general semi-algebraic formulas. Finally, as an applicaiton of our approach, we demonstrate how to apply it to invariant generation in program verification.
\end{abstract}



\section{Introduction} \label{sec:intr}

Interpolation-based techniques have become popular in recent years because of their inherently modular and local reasoning,
  which can scale up existing formal verification techniques like theorem proving, model-checking, abstract interpretation, and
  so on, while the scalability is the bottleneck of these techniques.
The study of interpolation was pioneered by Kraj${\rm \acute{i}\breve{c}}$ek \cite{krajicek97} and Pudl\'{a}k \cite{pudlak97} in
connection with theorem proving,
by McMillan in connection with model-checking \cite{mcmillan03},
by Graf and Sa\"{i}di \cite{GS97}, Henzinger et al. \cite{HJMM04} and McMillan \cite{mcmillan05} in connection with abstraction
 like CEGAR, by Wang et al. \cite{wang11} in connection with machine-learning based program verification.

Craig interpolant generation plays a central role in interpolation-based techniques, and therefore
  has drawn increasing attention. In the literature, there are various efficient algorithms proposed
    for automatically synthesizing  interpolants for various theories, e.g., decidable fragments of first-order logic, linear arithmetic, array logic, equality logic with uninterpreted functions (EUF), etc., and their combinations,
     and their use in verification {\cite{mcmillan05,HJMM04,YM05,KMZ06,RS10,KV09,CGS08,mcmillan08}.}
    In \cite{mcmillan05}, McMillan presented a method for deriving Craig interpolants from proofs in the quantifier-free theory of
 linear inequality and uninterpreted function symbols, and based on which an interpolating theorem prover was
 provided.
  In  \cite{HJMM04},   Henzinger \emph{et al.} proposed a method to synthesize Craig interpolants for a theory with arithmetic and pointer expressions, as well as call-by-value functions. In \cite{YM05}, Yorsh and Musuvathi  presented a combination method to generate Craig interpolants for a class of first-order
theories.  In \cite{KMZ06}, Kapur \emph{et al.} presented different efficient procedures to construct interpolants for
 the theories of arrays, sets and multisets using the reduction approach.  Rybalchenko and Sofronie-Stokkermans \cite{RS10} proposed an approach
 to reducing the synthesis of Craig interpolants of  the combined theory of linear arithmetic and
uninterpreted function symbols to constraint solving.
In addition, D'Silva et al. \cite{SPWK10} investigated
strengths of various interpolants.

However, interpolant generation for non-linear theory and its combination with the aforementioned theories is
    still in infancy,   although
nonlinear polynomials inequalities are quite common  in
 software involving number theoretic
functions as well as hybrid systems \cite{ZZKL12,Zhan17}.
In \cite{DXZ13},  Dai et al. had a first try and gave an
algorithm for generating interpolants for conjunctions of
mutually contradictory nonlinear polynomial inequalities
based on the existence of a witness guaranteed by Stengle's
Positivstellensatz \cite{Stengle}, which is computable using
semi-definite programming ($\sdp$).
Their algorithm is incomplete in general but if all variables are bounded (called Archimedean condition), then
their algorithm is complete. A major limitation of their work is
that two mutually contradictory formulas $\phi$ and $\psi$ must
have the same set of variables. \oomit{\footnote{See however an expanded version of their paper
  under preparation where they propose heuristics using program
  analysis for eliminating uncommon variables.}. }
In \cite{GDX16}, Gan \emph{et al.} proposed an algorithm to generate interpolants for
quadratic polynomial inequalities.
The basic idea is based on the insight that for analyzing the solution space
of concave quadratic polynomial inequalities, it
suffices to linearize them. A generalization of Motzkin's transposition theorem
is proved to be applicable for concave quadratic polynomial inequalities. Using this, they
proved  the existence of an interpolant for two mutually
contradictory conjunctions $\phi(\xx,\yy), \psi(\xx,\zz)$ of concave quadratic polynomial
inequalities and proposed an $\sdp$-based algorithm to compute it.
Also in \cite{GDX16}, they developed a combination algorithm for generating
interpolants for the combination of quantifier-free theory of concave quadratic polynomial
inequalities and \textit{EUF} based on the hierarchical calculus framework
proposed in \cite{SSLMCS2008} and used in \cite{RS10}.
\oomit{A polynomial is called
\emph{quadratic concave} iff $p(\xx) = \xx^T A \xx +\mathbf{\alpha}^T \xx + a$ with
$A$ is  semi-definite negative, where $\xx^T$ stands for the transposition of $\xx$, $A$ is a
 $|\xx|\times |\xx|$-dimension matrix over reals, and $\mathbf{\alpha}\in \mathbb{R}^{|\xx|}$.
 A polynomial formula is called \emph{quadratic concave} iff
all polynomials in it are \emph{quadratic concave}.} Obviously, \emph{quadratic concave} polynomial
inequalities is a very restrictive class of polynomial formulas, although most of existing abstract domains
fall within it as argued in \cite{GDX16}.  Meanwhile, in \cite{GZ16}, Gao and Zufferey presented an approach to extract interpolants for
non-linear formulas possibly containing transcendental functions and differential equations from
proofs of unsatisfiability generated by $\delta$-decision procedure \cite{Gao14} that are based on interval constraint propagation
(ICP) \cite{BG06}. Similar idea was also reported in \cite{KB11}. They transform proof traces from $\delta$-complete decision procedures into interpolants that consist of Boolean combinations of linear constraints. Thus, their approach can only find the interpolants between two formulas whenever their conjunction is not $\delta$-satisfiable.

\begin{example} \label{twist}
	
Let
{\small
\begin{eqnarray*}
\phi & = &  -2xy^2 +x^2-3xz-y^2 -yz+z^2-1\geq 0 \wedge
       100-x^2-y^2\geq 0 \wedge \\
       & &
     x^2z^2 +y^2z^2-x^2-y^2+ \frac{1}{6} (x^4+2x^2y^2+y^4)-\frac{1}{120}(x^6+y^6)-4 \leq 0; \\
\psi & = & 4(x-y)^4 +(x+y)^2 +w^2-133.097\leq 0
    \wedge 100(x+y)^2-w^2(x-y)^4-3000 \geq 0.
\end{eqnarray*}}
It can be checked that $\phi \wedge \psi \models \bot$.

\begin{center}
		\begin{figure}
			\flushleft
			\includegraphics[scale=0.5]{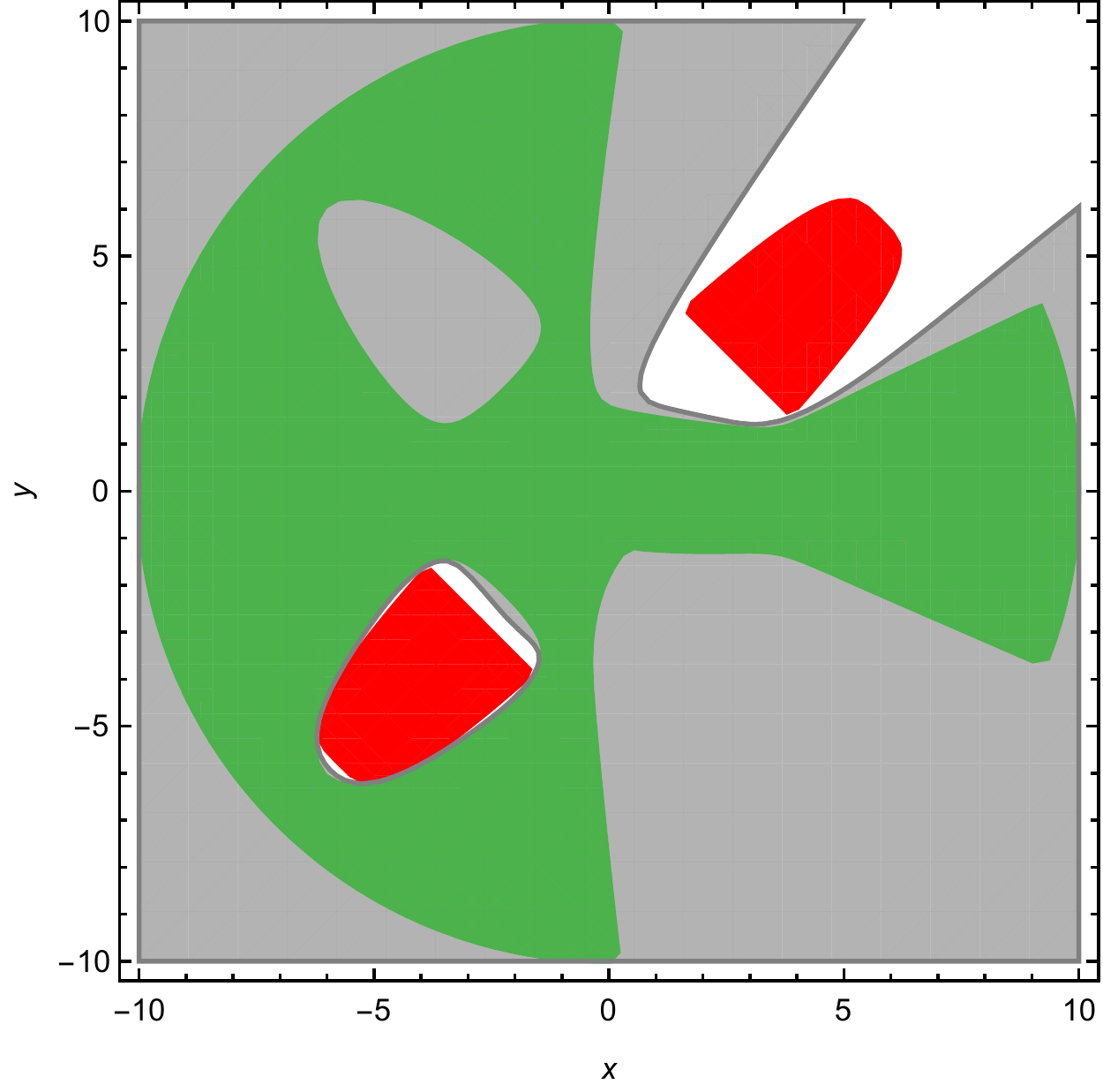}
			\caption{Example \ref{twist}. \small{(Green region: The projection of $\phi(x,y,z)$ onto $x$ and $y$; Red region: The projection of  $\psi(x,y,w)$ onto $x$ and $y$;
					 Gray region plus the green region:  The synthesized interpolant $\{(x,y)\mid h(x,y)>0\}$.)}}
			\label{fig-twist}
		\end{figure}
	\end{center}

Obviously, synthesizing
interpolants for $\phi$ and $\psi$ in this example is beyond the ability of the above approaches reported in \cite{DXZ13,GDX16}.
Using the method in \cite{GZ16} implemented in dReal3 
it would return ``SAT'' with $\delta=0.001$, i.e., $\phi\wedge \psi$ is $\delta$-satisfiable, and
 hence it cannot synthesize any interpolant using \cite{Gao14}'s approach with any precision greater than $0.001$\footnote{Alternatively, if we try the formula with the latest version of dReal4, it does not produce any output after 20 hours.}. While,
using our method, an interpolant  $h>0$ with degree $10$ can be found as shown in Fig \ref{fig-twist},
where,
{
\begin{align*}
  h=&85.868928829-3.48171895438*x-4.05055570642*y+69.5179717355*y^2-\\
  &1.975586161*x*y^2+73.7505185252*x^2+79.7889009541*x^4+\\
  &39.4019812625*x^2*y^2+59.0193160126*y^4+5.81006679403*x^6-\\
  &2.92654531095*y^6+43.1570807435*x*y+25.0366942193*x*y^3+\\
  &34.8268489109*x^3*y+3.86714864443*x^3+4.04378820386*x^2*y-\\
  &5.94528747172*y^3+8.72599707958*x^5+2.97413012798*x^4*y-\\
  &6.09661882807*x^3*y^2-6.92712190458*x^2*y^3-9.98542713025*x*y^4-\\
  &13.6495929388*y^5-28.5122945936*x^5*y-7.3306401143*x^4*y^2-\\
  &12.8405648769*x^3*y^3-2.4656615978*x^2*y^4-25.914058476*x*y^5+\\
  &0.574487452756*x^7-2.04161159811*x^6*y-2.1652665291*x^5*y^2+\\
  &0.441933622851*x^4*y^3-3.31339953201*x^3*y^4-1.18860782739*x^2*y^5-\\
  &1.59589948759*x*y^6-0.207834237584*y^7-0.39416212619*x^8-\\
  &2.61990676298*x^7*y+10.5074625167*x^6*y^2-11.2424047363*x^5*y^3+\\
  &4.40746440558*x^4*y^4-8.71563327623*x^3*y^5+12.7151449616*x^2*y^6-\\
  &7.65894409136*x*y^7+0.997807267855*y^8+0.360026919214*x^9-\\
  &1.60158863204*x^8*y+1.72767640966*x^7*y^2+1.39706253748*x^6*y^3-\\
  &3.00346730309*x^5*y^4+0.232691322651*x^4*y^5-1.70644550812*x^3*y^6+\\
  &5.74945017152*x^2*y^7-4.80737336486*x*y^8+1.1814541304*y^9+\\
  &0.387479209367*x^{10}-2.45667500053*x^9*y+7.28054219791*x^8*y^2-\\
  &12.6714503175*x^7*y^3+14.875324384*x^6*y^4-14.7134968452*x^5*y^5+\\
  &14.5995206041*x^4*y^6-12.7340406803*x^3*y^7+7.85021807731*x^2*y^8-\\
  &2.95705991155*x*y^9+0.522661022723*y^{10}.
\end{align*} }
Additionally, using the symbolic procedure REDUCE, it can be proved that $h>0$ is indeed an interpolant of $\phi$ and $\psi$.

\end{example}

In this paper, we further investigate this issue and consider how to synthesize
an interpolant for two polynomial formulas $\phi(\xx,\yy)$ and
$\psi(\xx,\zz)$ with $\phi(\xx,\yy)\wedge \psi(\xx,\zz) \models \bot$, where
\begin{center}
  $\phi(\xx,\yy) : f_1(\xx,\yy) \ge 0 \wedge \cdots \wedge f_m(\xx,\yy)\ge 0$, \\
  $\psi(\xx,\zz) : g_1(\xx,\zz) \ge 0 \wedge \cdots \wedge g_n(\xx,\zz)\ge 0$,
\end{center}
 $\xx \in \RR^r$, $\yy \in \RR^s$, $\zz \in \RR^t$ are variable vectors, $r,s,t \in \NN$,
and $f_1,\ldots,f_m,g_1,\ldots,g_n$ are polynomials. In addition,
$\MM_{\xx,\yy}\{ f_1(\xx,\yy),$ $\ldots, f_m(\xx,\yy) \}$ and
$\MM_{\xx,\zz}$ $\{ g_1(\xx,\zz)$, $\ldots$, $g_n(\xx,\zz) \}$ are two Archimedean quadratic modules
(the definition will be given later).
Here we allow uncommon variables, that are not allowed in \cite{DXZ13}, and drop the constraint
that polynomials must be concave and quadratic,
which is assumed in \cite{GDX16}.
The Archimedean condition amounts to that all the variables are bounded, which is reasonable
in program verification, as only bounded numbers can be represented in computer in practice. 
We first prove that there exists a polynomial $h(\xx)$ such that $h(\xx)=0$ separates the state space of $\xx$ defined by
$\phi(\xx,\yy)$ from that defined by $\psi(\xx,\zz)$ theoretically, and then propose an algorithm to compute
such $h(\xx)$ based on $\sdp$. 
Furthermore, we propose a verification approach to assure the validity of the synthesized interpolant and consequently avoid the unsoundness caused by numerical error in $\sdp$ solving. Finally, we also discuss how to extend our results to general semi-algebraic constraints.

Another contribution of this paper is that as an application, we illustrate how to apply our approach to invairant generation
in program verification by revising the framework proposed in  \cite{LSXLSH2017} by Lin \emph{et al.}  for invariant generation based on \emph{weakest precondition},
\emph{strongest postcondition} and \emph{interpolation}. It consists of two procedures, i.e.,
synthesizing invariants by forward interpolation based on \emph{strongest postcondition} and
\emph{interpolant generation}, and by backward interpolation  based on \emph{weakest precondition} and
\emph{interpolant generation}. In \cite{LSXLSH2017}, only linear invariants can be synthesized as no powerful
approaches are available to synthesize nonlinear interpolants. Obviously, our results can
strengthen their framework by allowing to generate nonlinear invariants.
To this end, we revise the two procedures in their
framework accordingly.



The paper is organized as follows. Some preliminaries are introduced in Section~\ref{sec:prel}.
Section~\ref{sec:existence} shows the existence of an interpolant for two mutually contradictory polynomial
 formulas only containing conjunction, and Section~\ref{sec:ssc} presents SDP-based methods to compute it.
  In Section~\ref{sec:numericalerror}, we discuss how to avoid unsoundness caused by numerical error in SDP.
  Section~\ref{sec:generalization} extends  our approach
 to general polynomial formulas. Section~\ref{sec:invariant} demonstrates how to apply our approach to invariant generation in program verification.
 We conclude this paper  in Section  \ref{sec:con}.

\section{Preliminaries} \label{sec:prel}
In this section, we first give a brief introduction on some notions used throughout the rest of this paper and then describe the problem of interest.
\subsection{Quadratic Module}
\label{sn}
$\mathbb{N}$, $\mathbb{Q}$ and $\mathbb{R}$ are the sets of integers, rational numbers and real numbers, respectively. $\mathbb{Q}[\xx]$ and $\mathbb{R}[\xx]$ denotes the polynomial ring over rational numbers and real numbers in $r\geq 1$ indeterminates $\xx:(x_1,\ldots,x_r)$. We use $\mathbb{R}[\xx]^2:=\{p^2\mid p\in \mathbb{R}[\xx]\}$ for the set of squares and $\sum\mathbb{R}[\xx]^2$ for the set of sums of squares of polynomials in $\xx$. Vectors are denoted by boldface letters. $\bot$ and $\top$ stand for \textbf{false} and \textbf{true}, respectively.


\begin{definition}[Quadratic Module \cite{Marshall2008}]
A subset $\mathcal{M}$ of $\mathbb{R}[\xx]$ is called a \emph{quadratic module} if it contains 1 and  is closed under addition and  multiplication with squares, i.e.,
\[ 1\in\mathcal{M}, \mathcal{M} + \mathcal{M} \subseteq \mathcal{M}, \mbox{ and } p^2 \mathcal{M} \subseteq \mathcal{M}~ \mbox{ for all } p \in \mathbb{R}[\xx].\]
\end{definition}


\begin{definition}
  Let $\overline{p}:=\{p_1,\ldots,p_s\}$ be a finite subset of $\mathbb{R}[\xx]$, the quadratic module $\mathcal{M}_{\xx}( \overline{p})$ or simply $\mathcal{M}(\overline{p})$ generated by $\overline{p}$ (i.e. the smallest quadratic module containing
  all $p_i$s) is
  \[\mathcal{M}_{\xx}(\overline{p})=\{\sum_{i=0}^s \delta_i p_i\mid  \delta_i \in \sum \mathbb{R}[\xx]^2\},\]
  where $p_0=1$.
\end{definition}

In other words, the quadratic module generated by $\overline{p}$ is a subset of polynomials that are nonnegative on the set $\{\xx\mid p_i(\xx)\geq 0, i=1,\ldots,s\}$. 
The following Archimedean condition plays a key role in the study of polynomial optimization.

\begin{definition}[Archimedean]
  Let $\mathcal{M}$ be a quadratic module of $\mathbb{R}[\xx]$ with $\xx=(x_1,\ldots,x_r)$. $\mathcal{M}$ is said to be \emph{Archimedean} if there exists some $a>0$ such that $a-\sum_{i=1}^{r} x_i^2 \in \mathcal{M}$.
\end{definition}

\subsection{Problem Description} \label{sec:prel:2}

Craig showed that given two formulas $\phi$ and $\psi$ in a
first-order theory $\TT$ s.t.
$\phi \models \psi$, there always exists an \emph{interpolant} $I$ over
the common symbols of $\phi$ and $\psi$ s.t.  $\phi \models
I, I \models \psi$. In the verification literature, this
terminology has been abused following \cite{mcmillan05}, where a
\emph{reverse interpolant} (coined by Kov\'{a}cs and Voronkov in \cite{KV09}) $I$ over the common symbols of $\phi$ and
$\psi$ is defined by
\begin{definition}[Interpolant]
\label{interpolant}
  Given two formulas $\phi$ and $\psi$ in a theory $\TT$ s.t. $\phi \wedge \psi \models_{\TT} \bot$,
  a formula $I$ is an \emph{interpolant} of $\phi$ and $\psi$ if (i) $\phi \models_{\TT} I$;
  (ii) $I \wedge \psi \models \bot$; and (iii) $I$ only contains common symbols and free variables
  shared by $\phi$ and $\psi$.
\end{definition}

\begin{definition}
\label{sets}
A basic semi-algebraic set $\{\xx\in \mathbb{R}^n\mid \bigwedge_{i=1}^s p_i(\xx)\geq 0\}$ is called a set of the {\em Archimedean form} if $\mathcal{M}_{\xx}\{p_1(\xx),\ldots,p_s(\xx)\}$ is Archimedean,
where $p_i(\xx)\in \mathbb{R}[\xx]$, $i=1,\ldots,s$.
\end{definition}

The interpolant synthesis problem of interest in this paper is described in \textbf{Problem} \ref{problem:1}.
\begin{problem} \label{problem:1}
Let $\phi(\xx,\yy)$ and $\psi(\xx,\zz)$ be two polynomial formulas defined as follows,
\begin{align*}
  \phi(\xx,\yy) : f_1(\xx,\yy) \geq 0 \wedge \cdots \wedge f_m(\xx,\yy)\ge 0, \\
  \psi(\xx,\zz) : g_1(\xx,\zz) \geq 0 \wedge \cdots \wedge g_n(\xx,\zz)\ge 0,
\end{align*}
where, $\xx \in \mathbb{R}^r$, $\yy \in \mathbb{R}^s$, $\zz \in \mathbb{R}^t$ are variable vectors, $r,s,t \in \mathbb{N}$,
and $f_1,\ldots,f_m,g_1$, \ldots, $g_n$ are polynomials in the corresponding variables.
Suppose $\phi \wedge \psi \models \bot$, and $\{(\xx,\yy)\mid \phi(\xx,\yy)\}$ and $\{(\xx,\zz)\mid \psi(\xx,\zz)\}$ are semi-algebraic sets of the Archimedean form. Find a polynomial $h(\xx)$ such that $h(\xx)>0$ is an interpolant for $\phi$ and $\psi$.
\end{problem}

\oomit{
\section{Interpolant Synthesis} \label{sec:intSyn}
In this section, we focus our attention on elucidating our method for synthesizing a polynomial $h(\xx)$ as presented in \textbf{Problem} \ref{problem:1}. We first prove the existence of such polynomial function $h(\xx)$ in subsection \ref{EI}, and then recast the problem of synthesizing such $h(\xx)$ as a semi-definite programming problem in subsection \ref{ssc}.}

\section{Existence of Interpolant}  \label{sec:existence}

The basic idea and steps of proving the existence of interpolant are as follows: Because an interpolant of $\phi$ and $\psi$ contains only the common symbols in $\phi$ and $\psi$, it is natural to consider the projections of the sets defined by $\phi$ and $\psi$ on $\xx$, i.e. $P_{\xx}(\phi(\xx,\yy)) := \{ \xx \mid \exists \yy. \, \phi(\xx,\yy)\}$ and $P_{\xx}(\psi(\xx,\zz)) := \{ \xx \mid \exists \zz.\,  \psi(\xx,\zz)\}$, which are obviously disjoint. We therefore prove that, if $h(\xx)=0$ separates $P_{\xx}(\phi(\xx,\yy))$ and $P_{\xx}(\psi(\xx,\zz))$, then $h(\xx)$ solves \textbf{Problem} \ref{problem:1} (see Proposition \ref{prop:1}). Thus, we only need to prove the existence of such $h(\xx)$ through the following steps.

First, we prove that $P_{\xx}(\phi(\xx,\yy))$ and $P_{\xx}(\psi(\xx,\zz))$ are compact semi-algebraic sets which are unions of finitely many basic closed semi-algebraic sets (see Lemma \ref{the:forproj}). Second, using Putinar's Positivstellensatz, we prove that, for two disjoint basic closed semi-algebraic sets $S_1$ and $S_2$ of the Archimedean form, there exists a polynomial $h_1(\xx)$ such that $h_1(\xx)=0$ separates $S_1$ and $S_2$ (see Lemma \ref{the:scase}). This result is then extended to the case that $S_2$ is a finite union of basic closed semi-algebraic sets (see Lemma \ref{lemma:1}). Finally, by generalizing Lemma \ref{lemma:1} to the case that two compact semi-algebraic sets both are unions of finitely many basic closed semi-algebraic sets ($P_{\xx}(\phi(\xx,\yy))$ and $P_{\xx}(\psi(\xx,\zz))$ are in this case by Lemma \ref{the:forproj}) and combining Proposition \ref{prop:1}, we prove the existence of interpolant in Theorem \ref{the:main} and Corollary \ref{conse}.

\begin{proposition}
\label{prop:1}
  If $h(\xx)\in \mathbb{R}[\xx]$ satisfies the following constraints
\begin{equation}
\label{formu:h}
\begin{split}
 &\forall \xx \in P_{\xx}(\phi(\xx,\yy)), h(\xx) > 0,\\
 &\forall \xx \in P_{\xx}(\psi(\xx,\zz)), h(\xx) < 0,
\end{split}
\end{equation}
then $h(\xx)> 0$ is an interpolant for $\phi(\xx,\yy)$ and $\psi(\xx,\zz)$, where $\phi(\xx,\yy)$ and $\psi(\xx,\zz)$ are defined as in \textbf{Problem} \ref{problem:1}.
\end{proposition}
\begin{proof}
 According to Definition \ref{interpolant}, it is enough to prove that $\phi(\xx,\yy) \models h(\xx)>0$ and $\psi(\xx,\zz) \models h(\xx) \leq 0$.

  Since any $(\xx_0,\yy_0)$ satisfying $\phi(\xx,\yy)$ must imply $\xx_0 \in P_{\xx}(\phi(\xx,\yy))$, it follows  that $h(\xx_0)>0$ from \eqref{formu:h} and $\phi(\xx,\yy) \models h(\xx)>0$. Similarly, we can prove $\psi(\xx,\zz) \models h(\xx)<0$, implying that $\psi(\xx,\zz) \models h(\xx) \leq 0$. Therefore, $h(\xx)> 0$ is an interpolant for $\phi(\xx,\yy)$ and $\psi(\xx,\zz)$.
\end{proof}

For the sake of synthesizing such $h(\xx)$ in Proposition \ref{prop:1}, we first dig deeper into the two sets $P_{\xx}(\phi(\xx,\yy))$ and $P_{\xx}(\psi(\xx,\zz))$.  As shown later, i.e. in Lemma \ref{the:forproj} ,  we will find that  these two sets are compact semi-algebraic sets of the form $\{\xx\mid\bigvee_{i=1}^{c} \bigwedge_{j=1}^{J_i} \alpha_{i,j}(\xx) \geq 0\}$. Before  this lemma,  we introduce Finiteness theorem pertinent to a {\em basic closed semi-algebraic subset} of $\RR^n$, which will be used in the proof of Lemma \ref{the:forproj}, where a basic closed semi-algebraic subset of $\RR^n$ is a set of the form $$\{\xx\in\RR^n\mid \alpha_1(\xx)\ge0,\ldots,\alpha_k(\xx)\ge0\}$$ with $\alpha_1,\ldots,\alpha_k\in\RR[\xx]$.

\begin{theorem}[Finiteness Theorem, Theorem 2.7.2 in \cite{BCR98}]
 \label{the:finiteness}
  Let $A\subset \RR^n$ be a closed semi-algebraic set. Then $A$ is a finite union of basic closed semi-algebraic sets.
\end{theorem}

\begin{lemma}
\label{the:forproj}
The set $P_{\xx}(\phi(\xx,\yy))$ is compact semi-algebraic set of the following form
   $$ P_{\xx}(\phi(\xx,\yy)):=\{\xx\mid\bigvee_{i=1}^{c} \bigwedge_{j=1}^{J_i} \alpha_{i,j}(\xx) \ge 0\},$$
   where $\alpha_{i,j}(\xx)\in \mathbb{R}[\xx]$, $i=1,\ldots,c$, $j=1,\ldots,J_i$. The same claim applies to the set $P_{\xx}(\psi(\xx,\zz))$ as well.
\end{lemma}
\begin{proof}[of Lemma~\ref{the:forproj}]
For the sake of simple exposition, we denote $\{(\xx,\yy)\mid \phi(\xx,\yy)\}$ and $ P_{\xx}(\phi(\xx,\yy))$ by $S$ and $\pi(S)$, respectively.

Because $S$ is a compact set, $\pi$ is a continuous map, and
continuous function maps compact set to compact set, then $\pi(S)$, which is the image of
a compact set under a continuous map, is compact. Moreover, since $S$ is a semi-algebraic set, and by Tarski-Seidenberg theorem \cite{Bierstone1988}  the projection of a semi-algebraic set is also
a semi-algebraic set,  this implies that  $\pi(S)$ is a semi-algebraic set. Thus, $\pi(S)$ is a compact semi-algebraic set.

Since $\pi(S)$ is a compact semi-algebraic set, and also a closed semi-algebraic set,  we have that
  $\pi(\SSS)$ is a finite union of basic closed semi-algebraic sets from Theorem \ref{the:finiteness}.
  Thus, there exist a series of polynomials
  $\alpha_{1,1}(\xx)$, \ldots, $\alpha_{1,J_1}(\xx)$, \ldots,
  $\alpha_{c,1}(\xx)$, \ldots, $\alpha_{c,J_c}(\xx)$ such that
   \begin{align*}
  \pi(\SSS)&=\bigcup_{i=1}^{c}\{\xx\mid \bigwedge_{j=1}^{J_i} \alpha_{i,j}(\xx) \ge 0\}
  =\{\xx\mid \bigvee_{i=1}^{c}\bigwedge_{j=1}^{J_i} \alpha_{i,j}(\xx) \ge 0\}.
  \end{align*}

  We conclude, we have proved this lemma.
\end{proof}

After knowing the structure of $P_{\xx}(\phi(\xx,\yy))$ and $P_{\xx}(\psi(\xx,\zz))$ being a union of some basic semialgebraic sets as illustrated in Lemma \ref{the:forproj}, we next prove the existence of $h(\xx)\in \mathbb{R}[\xx]$ satisfying \eqref{formu:h}, as formally stated in Theorem \ref{the:main}.
\begin{theorem}
\label{the:main}
Suppose that $\phi(\xx,\yy)$ and $\psi(\xx,\zz)$ are defined as in \textbf{Problem} 1. Then there exists a polynomial $h(\xx)$ satisfying \eqref{formu:h}.
\end{theorem}

A formal proof of Theorem \ref{the:main} requires some preliminaries, which will be given later. The main tool in our proof is Putinar's Positivstellensatz, as formulated in Theorem \ref{the:putinar}.
\begin{theorem}[Putinar's Positivstellensatz \cite{Putinar1993}]
\label{the:putinar}
  Let $p_1,\ldots,p_k \in \mathbb{R}[\xx]$ and $S_1=\{\xx\mid p_1(\xx)\geq 0, \ldots, p_k(\xx) \geq 0 \}$.
  Assume that the quadratic module $\mathcal{M}(p_1,\ldots,p_k)$ is Archimedean. For $q\in\mathbb{R}[\xx]$,
  if $q>0$ on $S_1$ then $q\in \mathcal{M}(p_1,\ldots,p_k)$.
\end{theorem}

 With Putinar's Positivstellensatz we can draw a conclusion that there exists a polynomial such that its zero level set\footnote{The zero level set of an $n$-variate polynomial $h(\xx)$ is defined as $\{\xx\in\mathbb{R}^n\mid h(\xx)=0\}.$} separates two compact semi-algebraic sets of the Archimedean form, as claimed in Lemmas \ref{the:scase} and \ref{lemma:1}. Theorem \ref{the:main} is a generalization of these two lemmas.
\begin{lemma}
\label{the:scase}
  Let $$S_1=\{\xx \mid p_1(\xx)\geq 0, \ldots, p_J(\xx)\geq 0\},$$
   $$S_2=\{\xx \mid q_1(\xx)\geq 0, \ldots, q_K(\xx)\geq 0\}$$ be semi-algebraic sets of the Archimedean form and $S_1 \cap S_2 = \emptyset$, 
  then there exists
  a polynomial $h_1(\xx)$ such that
  \begin{align}
    \forall \xx \in S_1,~~~ h_1(\xx)>0, \\
    \forall \xx \in S_2,~~~ h_1(\xx)<0.
  \end{align}
\end{lemma}
\begin{proof}[of Lemma~\ref{the:scase}]
  Since $S_1 \cap S_2 = \emptyset$, i.e.,
  $$p_1\geq 0\wedge\cdots\wedge p_J \geq 0 \wedge q_1\geq 0\wedge\cdots\wedge q_K\geq 0 \models \bot,$$
  it follows 
  $$p_2\geq 0\wedge\cdots\wedge p_J \geq 0 \wedge q_1\geq 0\wedge\cdots\wedge  q_K \geq 0 \models -p_1> 0.$$
  Let $S_3= \{\xx\mid p_2\geq 0\wedge\cdots\wedge p_J \geq 0 \wedge q_1\geq 0\wedge\cdots\wedge  q_K \geq 0\}$,
  then $-p_1>0$ on $S_3$. Since $S_1$ and $S_2$ are semi-algebraic sets of the Archimedean form, $$\mathcal{M}_{\xx}(p_2(\xx),\ldots,p_J(\xx),q_1(\xx), \ldots, q_K(\xx))$$ is also Archimedean,
   and thus $S_3$ is compact. From $-p_1>0$ on $S_3$,  we further have that there exists some $u_1 \in \sum \mathbb{R}[\xx]^2$
   such that $-u_1p_1-1>0$ on $S_3$. Using Theorem \ref{the:putinar}, we have that
  $$-u_1p_1-1 \in \mathcal{M}_{\xx}(p_2(\xx),\ldots,p_J(\xx),q_1(\xx), \ldots, q_K(\xx)),$$  implying that there exists a set of sums of squares polynomials $u_2,\ldots,u_J$ and $v_0$,$v_1$, \ldots, $v_K \in \mathbb{R}[\xx]$,
  such that
  $$-u_1p_1-1\equiv u_2p_2+\cdots+u_Jp_J+v_0+v_1q_1+\cdots+v_Kq_K.$$
  Let $h_1=\frac{1}{2}+u_1p_1+\cdots+u_Jp_J$, i.e.,
  $-h_1=\frac{1}{2}+v_0+v_1q_1+\cdots+v_Kq_K$.
  It is easy to check that
  \begin{align*}
    \forall \xx \in S_1,~~~ h_1(\xx)>0, \\
    \forall \xx \in S_2,~~~ h_1(\xx)<0.
  \end{align*}
\end{proof}
Lemma \ref{lemma:1} generalizes the results of Lemma \ref{the:scase} to more general compact semi-algebraic sets of the Archimedean form, which is the union of multiple basic semi-algebraic sets.
\begin{lemma}
\label{lemma:1}
Assume $S_0=\{\xx \mid p_1(\xx)\geq 0, \ldots, p_J(\xx)\ge 0\}$ and $S_i=\{\xx \mid q_{i,1}(\xx)\geq 0, \ldots, q_{i,K_i}(\xx)\geq 0\}$, $i=1,\ldots,b$, are semi-algebraic sets of the Archimedean form, and $S_0 \cap \bigcup_{i=1}^{b}S_i = \emptyset$, then there exists
  a polynomial $h_0(\xx)$ such that
  \begin{equation}
  \label{cond:l1}
  \begin{split}
    &\forall \xx \in S_0, ~~~h_0(\xx)>0, \\
    &\forall \xx \in \bigcup_{i=1}^{b}S_i,~~~ h_0(\xx)<0.
    \end{split}
\end{equation}
\end{lemma}

In order to prove this lemma, we prove the following lemma first.

\begin{lemma} \label{lemma:2}
  Let $c,d \in \mathbb{R}$ with $0<c<d$ and $U_0=[c,d]^r$.
  There exists a polynomial $\hat{h}(\xx)$ such that
  \begin{equation}
  \label{11}
    \xx\in U_0 \models \hat{h}(\xx) >0 \models \bigwedge_{i=1}^{r} x_i > 0,
  \end{equation}
where $\xx=(x_1,\ldots,x_r)$.
\end{lemma}
\begin{proof}[of Lemma~\ref{lemma:2}]
 We show that there exists $k \in \mathbb{N}$ such that
  $$\hat{h}(\xx)=(\frac{d}{2})^{2k}-(x_1-\frac{c+d}{2})^{2k}-\cdots-(x_r-\frac{c+d}{2})^{2k}$$ satisfies \eqref{11}.
It is evident that $\hat{h}(\xx)>0 \models \bigwedge_{i=1}^r x_i>0$ holds.  In the following we just need to verify that $\bigwedge_{i=1}^r c\leq x_i \le d \models  \hat{h}(\xx)>0$ holds.
  Since $c\le x_i \le d$, we have $(x_i-\frac{c+d}{2})^{2k} \le (\frac{d-c}{2})^{2k}$ and $$(\frac{d}{2})^{2k}-\sum_{i=1}^{r}(x_i-\frac{c+d}{2})^{2k}\geq (\frac{d}{2})^{2k}-r(\frac{d-c}{2})^{2k}.$$
Obviously, if an interger $k$ satisfies $(\frac{d}{d-c})^{2k}> r$, then
$(\frac{d}{2})^{2k}-\sum_{i=1}^{r}(x_i-\frac{c+d}{2})^{2k}>0$. The existence of such $k$ satisfying $(\frac{d}{d-c})^{2k}> r$ is assured by  $\frac{d}{d-c}> 1$.
\end{proof}

Now we give a proof of Lemma \ref{lemma:1} as follows.
\begin{proof}[of Lemma~\ref{lemma:1}]
 For any $i$ with $1\le i\le b$, according to Lemma \ref{the:scase}, there exists a
  polynomial $h_i\in \RR[\xx]$, satisfying
   \begin{align*}
    \forall \xx \in S_0,~~~ h_i(\xx)>0, \\
    \forall \xx \in S_i,~~~ h_i(\xx)<0.
  \end{align*}

Next, we construct $h_0(\xx)\in \mathbb{R}[\xx]$ from $h_1(\xx), \ldots, h_b(\xx)$. Since $S_0$ is a semi-algebraic set of the Archimedean form, $S_0$ is compact and thus $h_i(\xx)$ has minimum value and maximum value on $S_0$, denoted by $c_i$ and $d_i$. Let $c=\min(c_1, \ldots, c_b)$ and $d=\max(d_1, \ldots, d_b)$. It is evident that $0<c<d$.

  From Lemma \ref{lemma:2} there must exist a polynomial $\hat{h}(w_1,\ldots,w_b)$ such that
  \begin{align}
    \bigwedge_{i=1}^b c\le w_i \le d \models \hat{h}(w_1,\ldots,w_b)>0, \label{for:l1:3} \\
    \hat{h}(w_1,\ldots,w_b)>0 \models \bigwedge_{i=1}^b w_i>0. \label{for:l1:4}
  \end{align}
  Let $h'_0(\xx)=\hat{h}(h_1(\xx), \ldots,h_b(\xx))$. Obviously, $h'_0(\xx)\in \mathbb{R}[\xx]$. We next prove that $h'_0(\xx)$ satisfies \eqref{cond:l1} in Lemma \ref{lemma:1}.

  For all $\xx_0\in S_0$,  $c \leq h_i(\xx_0) \leq d$, $i=1,\ldots,b$,  $h'_0(\xx_0)=\hat{h}(h_1(\xx_0), \ldots,h_b(\xx_0))$ $ >0$ from formula (\ref{for:l1:3}). Therefore, the first constraint in \eqref{cond:l1}, i.e. $\forall \xx_0\in S_0, h_0(\xx_0)$ $>0$, holds.

  For any $\xx_0\in \bigcup_{i=1}^{b} S_i$, there must exist some $i$ such that $\xx_0 \in S_i$, implying that $h_i(\xx_0)<0$. From formula (\ref{for:l1:4}) we have $h'_0(\xx_0)=\hat{h}(h_1(\xx_0), \ldots,h_b(\xx_0)) \leq 0$.

  Thus, we obtain the conclusion that there exists a polynomial $h'_0(\xx)$ such that
  \begin{align*}
    \forall \xx \in S_0, ~~~h'_0(\xx)>0, \\
    \forall \xx \in \bigcup_{i=1}^{b}S_i,~~~ h'_0(\xx)\le 0.
  \end{align*}
 Also, since $S_0$ is a compact set, and $h'_0(\xx)>0$ on $S_0$, there must exist some positive number $\epsilon >0$
  such that $h'_0(\xx)-\epsilon>0$ over $S_0$. Then $h'_0(\xx)-\epsilon < 0$ on $\bigcup_{i=1}^{b}S_i$.
  Therefore, setting $h_0(\xx):=h'_0(\xx)-\epsilon$, the conclusion in Lemma \ref{lemma:1} is proved.
\end{proof}

In Lemma \ref{lemma:1} we proved that there exists a polynomial $h(\xx) \in \mathbb{R}[\xx]$ such that its zero level set is a barrier between two semi-algebraic sets of the Archimedean form, of which one set is a union of finitely many basic semi-algebraic sets. In the following we will give a formal proof of Theorem \ref{the:main}, which is a generalization of Lemma \ref{lemma:1} by considering the situation that two compact semi-algebraic sets both are unions of finitely many basic semi-algebraic sets.
\begin{proof}[Proof of Theorem~\ref{the:main}]
  According to Lemma \ref{the:forproj} we have that $P_{\xx}(\phi(\xx,\yy))$ and $P_{\xx}(\psi(\xx,\zz))$ are compact sets, and there respectively exists a set of polynomials
  $p_{i,j}(\xx)\in \mathbb{R}[\xx]$, $i=1,\ldots,a$, $j=1,\ldots,J_i$, and $q_{l,k}(\xx)\in \mathbb{R}[\xx]$, $l=1,\ldots,b$, $k=1,\ldots,K_i$,
  such that
  \begin{align*}
    P_{\xx}(\phi(\xx,\yy))=\{\xx\mid\bigvee_{i=1}^{a} \bigwedge_{j=1}^{J_i} p_{i,j}(\xx) \geq 0\}, \\
    P_{\xx}(\psi(\xx,\zz))=\{\xx\mid\bigvee_{l=1}^{b} \bigwedge_{k=1}^{K_l} q_{l,k}(\xx) \geq 0\}.
  \end{align*}
Since $P_{\xx}(\phi(\xx,\yy))$ and $P_{\xx}(\psi(\xx,\zz))$ are compact sets, there exists a positive $N\in \mathbb{R}$ such that $f=N-\sum_{i=1}^rx_i^2\ge0$ over $P_{\xx}(\phi(\xx,\yy))$ and $P_{\xx}(\psi(\xx,\zz))$. For each $i=1,\ldots,a$ and each $l=1,\ldots,b$, set $p_{i,0}=q_{l,0}=f$. Denote
$$\{\xx\mid\bigvee_{i=1}^{a} \bigwedge_{j=0}^{J_i} p_{i,j}(\xx) \geq 0\}=\bigcup_{i=1}^a\{\xx\mid \bigwedge_{j=0}^{J_i} p_{i,j}(\xx) \geq 0\}$$ by $P_1$ and
$$\{\xx\mid\bigvee_{l=1}^{b} \bigwedge_{k=0}^{K_l} q_{l,k}(\xx) \geq 0\}=\bigcup_{l=1}^b \{\xx\mid\bigwedge_{k=0}^{K_l} q_{l,k}(\xx) \geq 0\}$$ by $P_2$. It is easy to see that $P_1=P_{\xx}(\phi(\xx,\yy)$, $P_2=P_{\xx}(\psi(\xx,\zz))$.


 Since $\phi \wedge \psi \models \bot$, there does not exist
  $(\xx,\yy,\zz)\in \mathbb{R}^{r+s+t}$ that satisfies $\phi \wedge \psi$, implying that $P_{\xx}(\phi(\xx,\yy)) \cap P_{\xx}(\psi(\xx,\zz))=\emptyset$ and thus $P_1\cap P_2=\emptyset$. Also, since $\{ \xx\mid \bigwedge_{j=0}^{J_{i_1}} p_{i_1,j}(\xx) \geq 0\}\subseteq P_1$, for each $i_1=1,\ldots,a$,
  $$\{\xx\mid \bigwedge_{j=0}^{J_{i_1}} p_{i_1,j}(\xx) \geq 0\} \cap P_2=\emptyset$$ holds. From Lemma \ref{lemma:1} there exists $h_{i_1}(\xx)\in\mathbb{R}[\xx]$ such that
  \begin{align*}
    &\forall \xx \in \{ \xx\mid \bigwedge_{j=0}^{J_{i_1}} p_{i_1,j}(\xx) \geq 0\} \Rightarrow h_{i_1}(\xx)>0, \\
    &\forall \xx \in P_2 \Rightarrow h_{i_1}(\xx)<0.
  \end{align*}


  Let
  \begin{align*}
    S'=&\{\xx\mid -h_1(\xx)\ge 0, \ldots, -h_a(\xx)\ge 0, N-\sum_{i=1}^rx_i^2\geq 0\}.
  \end{align*}
Obviously, $S'$ is a semialgebraic set of the Archimedean form,
  $P_2\subset S'$ and
  $P_1 \cap S'=\emptyset.$  Therefore, according to Lemma \ref{the:scase}, there exists a polynomial $\overline{h}(\xx)\in \mathbb{R}[\xx]$ such that
  \begin{align*}
    &\forall \xx \in S',~~ \overline{h}(\xx) >0, \\
    &\forall \xx \in P_1, ~~\overline{h}(\xx)<0.
  \end{align*}
  Let $h(\xx)=-\overline{h}(\xx)$, then we have
  \begin{align*}
    &\forall \xx \in P_1, ~~h(\xx)>0, \\
    &\forall \xx \in P_2, ~~h(\xx)<0,
  \end{align*}
implying that
  \begin{align*}
  \forall \xx \in P_{\xx}(\phi(\xx,\yy)), h(\xx) > 0,  \\
  \forall \xx \in P_{\xx}(\psi(\xx,\zz)), h(\xx) < 0.
\end{align*}
Thus, we have proved Theorem \ref{the:main}.
\end{proof}

Consequently, we immediately have the following conclusion.
\begin{corollary}
\label{conse}
  Let $\phi(\xx,\yy)$ and $\psi(\xx,\zz)$ be defined as in \textbf{Problem} 1. There must exist a polynomial $h(\xx)\in\RR[\xx]$ such that $h(\xx)>0$ is an interpolant for $\phi$ and $\psi$.
\end{corollary}

Actually, since $P_{\xx}(\phi(\xx,\yy))$ and $P_{\xx}(\psi(\xx,\zz))$ both are compact set from
Lemma \ref{the:finiteness}, and $h(\xx)>0$ on $P_{\xx}(\phi(\xx,\yy))$, $h(\xx)<0$ on
$P_{\xx}(\psi(\xx,\zz))$, for a small perturbation of the coefficients of $h(\xx)$ to obtain $h'(\xx)$, $h'(\xx)$ should also has the property as $h(\xx)$. Thus, there should exists a $h(\xx) \in \mathbb{Q}[\xx]$ such that $h(\xx)>0$ is an interpolant for $\phi$ and $\psi$, intuitively.
We show this in the following theorem.

\begin{theorem}
  \label{rational-int}
  Let $\phi(\xx,\yy)$ and $\psi(\xx,\zz)$ be defined as in \textbf{Problem} 1. There must exist a polynomial $h(\xx)\in\mathbb{Q}[\xx]$ such that $h(\xx)>0$ is an interpolant for $\phi$ and $\psi$.
\end{theorem}

\begin{proof}[of Theorem~\ref{rational-int}]
  We just need to prove there exists a polynomial $h(\xx)\in\mathbb{Q}[\xx]$ satisfying formula (\ref{formu:h}).

  From Theorem \ref{the:main}, there exists a polynomial $h'(\xx)\in\RR[\xx]$ satisfying formula (\ref{formu:h}). Since $P_{\xx}(\phi(\xx,\yy))$ and $P_{\xx}(\psi(\xx,\zz))$ are compact sets, $h'(\xx)>0$ on $P_{\xx}(\phi(\xx,\yy))$ and $h'(\xx)<0$ on
$P_{\xx}(\psi(\xx,\zz))$, there exist $\eta_1>0$ and $\eta_2>0$ such that
\begin{align*}
 \forall \xx \in P_{\xx}(\phi(\xx,\yy)), h'(\xx)-\eta_1 \ge 0,~
 \forall \xx \in P_{\xx}(\psi(\xx,\zz)), h'(\xx)+\eta_2 \le 0.
\end{align*}
Let $\eta=\min(\frac{\eta_1}{2},\frac{\eta_2}{2})$.
Suppose $h'(\xx)\in\RR[\xx]$ has the following form
\begin{align*}
  h'(\xx)=\sum_{\alpha\in\Omega} c_{\alpha}\xx^{\alpha},
\end{align*}
where $\alpha\in \NN^r$, $\Omega \subset \NN^r$ is a finite set of indices,  $r$ is the dimension of $\xx$, $\xx^{\alpha}$ is the monomial $\xx_1^{\alpha_1}\cdots\xx_r^{\alpha_r}$, and $0\neq c_{\alpha}\in\RR$ is the coefficient of monomial $\xx^{\alpha}$. Let $N=|\Omega|$ be the cardinality of $\Omega$. Since $P_{\xx}(\phi(\xx,\yy))$ and $P_{\xx}(\psi(\xx,\zz))$ are compact sets, for any $\alpha\in\Omega$, there exists $M_{\alpha}>0$ such that
\begin{align*}
  M_{\alpha}=\max\{|\xx^{\alpha}| \mid \xx\in P_{\xx}(\phi(\xx,\yy)) \cup P_{\xx}(\psi(\xx,\zz))\}.
\end{align*}
Then for any fixed polynomial
\begin{align*}
  \hat{h}(\xx)=\sum_{\alpha\in\Omega} d_{\alpha}\xx^{\alpha},
\end{align*}
with $d_{\alpha}\in [c_{\alpha}-\frac{\eta}{NM_{\alpha}},c_{\alpha}+\frac{\eta}{NM_{\alpha}}]$,
and any $\xx \in P_{\xx}(\phi(\xx,\yy)) \cup P_{\xx}(\psi(\xx,\zz))$, we have
{\small
\begin{align*}
  |\hat{h}(\xx)-h'(\xx)|&=|\sum_{\alpha\in\Omega} (d_{\alpha}-c_{\alpha})\xx^{\alpha}|
  \le \sum_{\alpha\in\Omega} |(d_{\alpha}-c_{\alpha})|\cdot|\xx^{\alpha}|
  \le \sum_{\alpha\in\Omega}\frac{\eta}{NM_{\alpha}} \cdot M_{\alpha}=\eta.
\end{align*}}
Since $\eta = \min(\frac{\eta_1}{2},\frac{\eta_2}{2})$, hence
\begin{equation}
\label{formu:hhh}
\begin{split}
 &\forall \xx \in P_{\xx}(\phi(\xx,\yy)), \hat{h}(\xx)\ge\frac{\eta_1}{2}  > 0,\\
 &\forall \xx \in P_{\xx}(\psi(\xx,\zz)), \hat{h}(\xx) \le-\frac{\eta_2}{2} < 0.
\end{split}
\end{equation}
Since for any $d_\alpha \in [c_{\alpha}-\frac{\eta}{NM_{\alpha}},c_{\alpha}+\frac{\eta}{NM_{\alpha}}]$ the above formula
(\ref{formu:hhh}) holds, there must exist some rational number $r_{\alpha}\in\mathbb{Q}$ in
$[c_{\alpha}-\frac{\eta}{NM_{\alpha}},c_{\alpha}+\frac{\eta}{NM_{\alpha}}]$ satisfying (\ref{formu:hhh}) because of
the density of rational numbers. Thus,  let
 \[ h(\xx)=\sum_{\alpha\in\Omega} r_{\alpha}\xx^{\alpha}.\]
Clearly, it follows that $h(\xx)\in\mathbb{Q}[\xx]$ and formula (\ref{formu:h}) holds.
\end{proof}

Therefore, we proved the existence of $h(\xx)\in\mathbb{Q}[\xx]$. Moreover, from the proof of Theorem \ref{rational-int},
 we know that a small perturbation of $h(\xx)$ is permitted, which is a good property for computing $h(\xx)$ in a numeric way. In the subsequent subsection, we recast the problem of finding such $h(\xx)$ as a semi-definite programming problem.

\section{ SOS Formulation}
\label{sec:ssc}
Similar as in \cite{DXZ13}, in this section, we discuss how to reduce the problem of finding $h(\xx)$ satisfying \eqref{formu:h} to a sum of squares
programming problem, which falls within the convex programming framework, and therefore can be solved by interior-point methods efficiently.

\begin{theorem}
\label{the:main2}
  Let $\phi(\xx,\yy)$ and $\psi(\xx,\zz)$ be defined as in the \textbf{Problem} 1. Then there exist $m+n+2$ SOS (sum of squares) polynomials $u_i(\xx,\yy)~ (i=1,\ldots,m+1),$ $v_j(\xx,\zz)~ (j=1,\ldots,n+1)$ and a polynomial $h(\xx)$ such that
  \begin{align}
  \label{formula:main}
   &h-1=\sum_{i=1}^m u_if_i+u_{m+1},\\
   &-h-1=\sum_{j=1}^nv_jg_j+ v_{n+1},
  \end{align}
and $h(\xx)>0$ is an interpolant for $\phi(\xx,\yy)$ and $\psi(\xx,\zz)$.
\end{theorem}
\begin{proof}[of Theorem~\ref{the:main2}]

  From Theorem \ref{the:main} there exists a polynomial $\hat{h}(\xx)$ such that
  \begin{align*}
    \forall \xx \in P_{\xx}(\phi(\xx,\yy)), \hat{h}(\xx) > 0,  \\
    \forall \xx \in P_{\xx}(\psi(\xx,\zz)), \hat{h}(\xx) < 0.
  \end{align*}
  Set $S_1=\{(\xx,\yy)\mid f_1\ge 0, \ldots, f_m\ge 0\}$ and
  $S_2=\{(\xx,\zz)\mid g_1\ge 0, \ldots, g_n\ge 0\}$. Since $\hat{h}(\xx) > 0$ on $S_1$, which is compact, there exist $\epsilon_1>0$ such that $\hat{h}(\xx)-\epsilon_1 > 0$ on $S_1$.
  For the same reason, there exist $\epsilon_2>0$ such that $-\hat{h}(\xx)-\epsilon_2 > 0$ on $S_2$.
  Let $\epsilon=\min(\epsilon_1,\epsilon_2)$, and $h(\xx)=\frac{\hat{h}(\xx)}{\epsilon}$, then
  ${h}(\xx)-1 > 0$ on $S_1$ and $-{h}(\xx)-1 > 0$ on $S_2$.
  Since $\mathcal{M}_{\xx,\yy}(f_1(\xx,\yy),$ $\ldots,$ $f_m(\xx,\yy))$  is Archimedean, according to Theorem \ref{the:putinar}, we have
  $${h}(\xx) -1 \in \mathcal{M}_{\xx,\yy}(f_1(\xx,\yy),\ldots, f_m(\xx,\yy)). $$ Similarly,
  $$-{h}(\xx) -1 \in \MM_{\xx,\zz}(g_1(\xx,\zz),\ldots, g_n(\xx,\zz)). $$
  That is,  there exist $m+n+2$ SOS polynomials $u_i,v_j$ satisfying the following semi-definite constraints:
  \begin{align*}
   h(\xx)-1=\sum_{i=1}^m u_if_i+u_{m+1},\\
   -h(\xx)-1=\sum_{j=1}^nv_jg_j+ v_{n+1}.
  \end{align*}
\end{proof}

According to Theorem \ref{the:main2}, the problem of finding $h(\xx)\in \mathbb{R}[\xx]$ solving \textbf{Problem} 1 can be equivalently reformulated as the problem of searching for SOS polynomials $u_1(\xx,\yy),\ldots,u_{m}(\xx,\yy)$, $v_1(\xx,\zz),\ldots,v_{n}(\xx,\zz)$ and a polynomial $h(\xx)$ with  appropriate degrees such that
\begin{equation}
\label{sos} \left\{
\begin{split}
   &h(\xx)-1-\sum_{i=1}^m u_if_i \in \sum \mathbb{R}[\xx,\yy]^2,\\
   &-h(\xx)-1-\sum_{j=1}^nv_jg_j\in \sum \mathbb{R}[\xx,\zz]^2,\\
   &u_i \in \sum\mathbb{R}[\xx,\yy]^2, i=1,\ldots,m,\\
   &v_j \in \sum\mathbb{R}[\xx,\zz]^2, j=1,\ldots,n.
\end{split} \right.
\end{equation}
\eqref{sos} is SOS constraints over SOS multipliers $u_1(\xx,\yy),$ $\ldots,u_{m}(\xx,\yy)$, $v_1(\xx,\zz),$ \ldots, $v_{n}(\xx,\zz)$, polynomial $h(\xx)$, which is convex and could be solved by many existing semi-definite programming solvers such as the optimization library  AiSat \cite{DXZ13} built on CSDP \cite{CSDP}. Therefore, according to Theorem \ref{the:main2}, $h(\xx)>0 $ is an interpolant for $\phi$ and $\psi$, which is formulated in Theorem \ref{coro-inter}.
\begin{theorem}[Soundness]
\label{coro-inter}
Suppose that $\phi(\xx,\yy)$ and $\psi(\xx,\zz)$ are defined as in \textbf{Problem} \ref{problem:1}, and $h(\xx)$ is a feasible solution to \eqref{sos}, then $h(\xx)$ solves \textbf{Problem} \ref{problem:1}, i.e. $h(\xx)>0$ is an interpolant for $\phi$ and $\psi$.
\end{theorem}

Moreover, we have the following completeness theorem stating that if the degrees of the polynomial $h(\xx)\in \mathbb{R}[\xx]$ and sum of squares polynomials $u_i(\xx,\yy)\in \sum \mathbb{R}[\xx,\yy]^2,v_j(\xx,\zz)\in \sum\mathbb{R}[\xx,\zz]^2$, $i=1,\ldots,m$, $j=1,\ldots,n$, are large enough, $h(\xx)$ can be synthesized definitely via solving \eqref{sos}.
\begin{theorem}[Completeness] \label{the:completeness}
For \textbf{Problem} \ref{problem:1}, there must be polynomials $u_i(\xx,\yy)\in \mathbb{R}_N[\xx,\yy]$ $(i=1,\ldots,m)$, $v_j(\xx,\zz)\in \mathbb{R}_N[\xx,\zz]$ $(j=1,\ldots,n)$ and $h(\xx)\in \mathbb{R}_N[\xx]$ satisfying (11) for some positive integer $N$, where $\mathbb{R}_k[\cdot]$ stands for the family of polynomials of degree no more than $k$.
\end{theorem}
\begin{proof}[of Theorem~\ref{the:completeness}]
This is an immediate result of Theorem \ref{the:main2}.
\end{proof}

\begin{example}
\label{ex3}

\begin{figure}[htbp]
  \centering

  \includegraphics[scale=0.6]{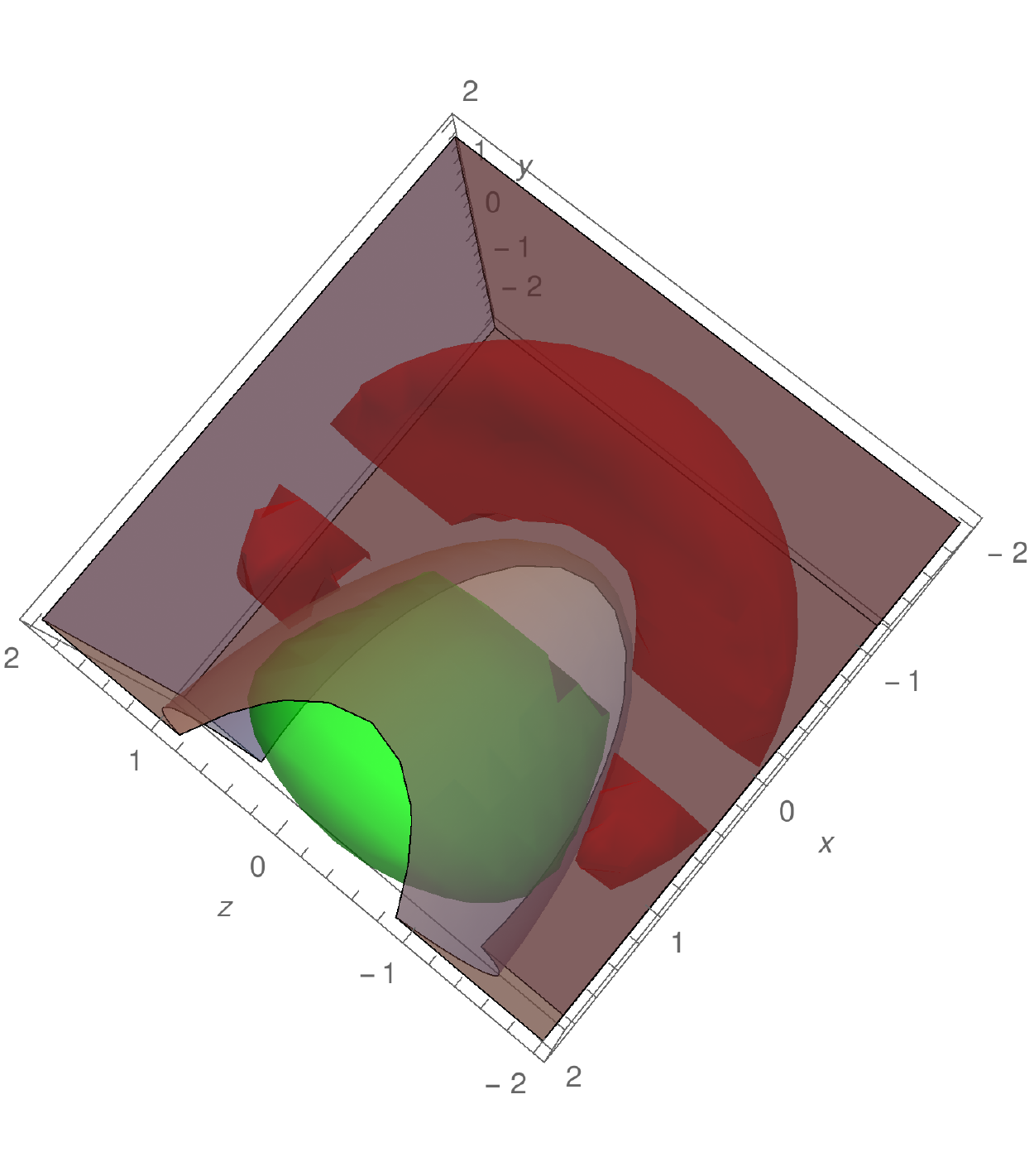}
  \caption{Example \ref{ex3}. \small{(Red region: $P_{x,y,z}(\phi(x,y,z,a_1,b_1,c_1,d_1))$; Green region:   $P_{x,y,z}(\psi(x,y,z,a_2,b_2,c_2,d_2))$; Gray region:  $\{(x,y,z)\mid h(x,y,z)>0\}$.)}}
  \label{fig-one-3}
\end{figure}

Consider two contradictory formulas $\phi$ and $\psi$ as follows:
$\phi(x,y,z$, $a_1$, $b_1,c_1,d_1):$
\begin{align*}
 &f_1(x,y,z,a_1,b_1,c_1,d_1)\geq 0 ~~\wedge\\ &f_2(x,y,z,a_1,b_1,c_1,d_1)\geq 0 ~~\wedge \\
 &f_3(x,y,z,a_1,b_1,c_1,d_1)\geq 0
\end{align*}
 and $\psi(x,y,z,a_2,b_2,c_2,d_2):$
\begin{align*}
 &g_1(x,y,z,a_2,b_2,c_2,d_2)\geq 0 ~~\wedge \\
  &g_2(x,y,z,a_2,b_2,c_2,d_2)\geq 0 ~~\wedge \\
  &g_3(x,y,z,a_2,b_2,c_2,d_2)\geq 0,
\end{align*}
where
  \begin{align*}
    &f_1=4-x^2-y^2-z^2-a_1^2-b_1^2-c_1^2-d_1^2,\\
    &f_2=-y^4+2x^4-a_1^4-1/100,\\
    &f_3=z^2-b_1^2-c_1^2-d_1^2-x-1,\\
    &g_1=4-x^2-y^2-z^2-a_2^2-b_2^2-c_2^2-d_2^2,\\
    &g_2=x^2-y-a_2-b_2-d_2^2-3,\\
    &g_3=x.
  \end{align*}
It is easy to observe that $\phi$ and $\psi$ satisfy the conditions in \textbf{Problem} 1.
Since there are local variables in $\phi$ and $\psi$ and the degree of $f_2$ is $4$, the interpolant generation methods in \cite{DXZ13} and \cite{GDX16} are not applicable. We get a concrete $\sdp$ problem of the form (\ref{sos}) by setting the degree of the polynomial $h(x,y,z)$ in \eqref{sos} to be $2$. Using the MATLAB package YALMIP\footnote{It can be downloaded from \url{https://yalmip.github.io/}.} \cite{lofberg2004} and Mosek\footnote{For academic use, the software Mosek can be obtained free from \url{https://www.mosek.com/}.} \cite{mosek2015mosek}, we obtain
\begin{align*}
&h(x,y,z)=-416.7204-914.7840x+472.6184y\\
&\quad+199.8985x^2+
190.2252y^2+690.4208z^2-
187.1592xy.
\end{align*}
Pictorially, we plot $P_{x,y,z}(\phi(x,y,z,a_1,b_1,c_1,d_1))$, $P_{x,y,z}(\psi(x,y,z,a_2,b_2,c_2,$ $d_2))$ and $\{(x,y,z)\mid h(x,y,z)>0\}$ in Fig. \ref{fig-one-3}. 
It is evident that $h(x,y,z)$ as presented above for $d_h=2$ is a real interpolant for $\phi(x,y,z,a,b,c,d)$ and $\psi(x,y,z,a,b,c,d)$.

\end{example}

\section{Avoidance of the unsoundness due to numerical error in SDP} \label{sec:numericalerror}
To the best of our knowledge, all the efficient SDP solvers are based on interior point method,
which is a numerical method. Thus, the numerical error is inevitable in our approach.
In this section, we discuss how to avoid the unsoundness of our approach
caused by numerical error in SDP
 based on the work in \cite{RVS16}.

We say a square matrix $A$ is positive semidefinite is $A$ is real symmetric and all eigenvalues of $A$ are $\ge 0$, denote
$A \succeq 0$ for $A$ is positive semidefinite.

In order to solve formula (\ref{sos}) to obtain $h(\xx)$, we first need to fix a degree
bound of $u_i$, $v_j$ and $h$, say $2d$, $d\in\NN$.
It is well-known that any  $u(\xx) \in\sum \RR[\xx]^2$ with degree $2d$
can be represented by 
\begin{align}
\label{polytomatrix}
u(\xx) \equiv E_d(\xx)^T C_u E_d(\xx),
\end{align}
where $C_u\in\RR^{\binom{r+d}{d}\times\binom{r+d}{d}}$ with
$C_u\succeq 0$, $E_d(\xx)$ is a column vector with all monomials
in $\xx$, whose total degree is not greater  than $d$, and $E_d(\xx)^T$
stands for the transposition of $E_d(\xx)$.
Equaling the corresponding coefficient of each monomial whose degree
is less than or equal to $2d$ at the two sides of
 (\ref{polytomatrix}), we can get a  linear equation system of the form
\begin{align}
  \label{eqsforu}
\mathtt{tr}(A_{u,k}C_u)=b_{u,k},~ k=1,\ldots,K_u,
\end{align}
where $A_{u,k}\in\RR^{\binom{r+d}{d}\times\binom{r+d}{d}}$ is constant matrix, $b_{u,k}\in\RR$ is constant,
$\mathtt{tr}(A)$ stands for  the trace of matrix $A$.
Thus, searching for $u_i$, $v_j$ and $h$ satisfying (\ref{sos})
can be reduced to the following SDP problem:
\oomit{there are matrixes $C_{u_i}$, $C_{v_j}$, $C_h$
and linear equations corresponding to them, respectively.
Then the formula (\ref{sos}) can be translated to a SDP problem as:}
\begin{equation}
\label{sdp}
\begin{split}
   \mathtt{find} :~ &C_{u_1},\ldots,C_{u_m},C_{v_1},\ldots,C_{v_n},C_{h},\\
   \hspace{-1cm}\mathtt{s.t.}~&\mathtt{tr}(A_{u_i,k}C_{u_i})=b_{u_i,k},~i=1,\ldots,m, k=1,\ldots,K_{u_i},\\
   &\mathtt{tr}(A_{v_j,k}C_{v_j})=b_{v_j,k},~j=1,\ldots,n, k=1,\ldots,K_{v_j},\\
   &\mathtt{tr}(A_{h,k}C_h)=b_{h,k},~ k=1,\ldots,K_h,\\
   &\hspace{-1.3cm}\mathtt{diag}(C_{u_1},\ldots,C_{u_m},C_{v_1},\ldots,C_{v_n},C_{h-1-uf},C_{-h-1-vg})\succeq 0,
\end{split}
\end{equation}
where $C_{h-1-uf}$ is the matrix corresponding to polynomial $h-1-\sum_{i=1}^m u_if_i$, which is
a linear combination of $C_{u_1}$, \ldots, $C_{u_m}$ and $C_h$; similarly,
$C_{-h-1-vg}$ is the matrix corresponding  to polynomial $-h-1-\sum_{j=1}^n v_jg_j$, which is
a linear combination of $C_{v_1}$, \ldots, $C_{v_n}$ and $C_h$; and
$\mathtt{diag}(C_1,\ldots,C_k)$ is a block-diagonal matrix of $C_1, \ldots, C_k$.

Let $D$ be the dimension of $C=\mathtt{diag}(C_{u_1},\ldots,$ $C_{-h-1-vg})$, i.e., $\mathtt{diag}(C_{u_1}$, \ldots, $C_{-h-1-vg})\in\RR^{D\times D}$ and
$\widehat{C}$ be the approximate solution to (\ref{sdp}) returned by calling  a numerical SDP solver,  
 the following theorem is proved in \cite{RVS16}.
\begin{theorem} [\cite{RVS16}, Theorem 3]  \label{the:sdp:cond}
  $C\succeq 0$ if there exists $\widetilde{C}\in\mathbb{F}^{D\times D}$ such that the following conditions hold:
  \begin{itemize}
    \item[1.] $\widetilde{C}_{ij}=C_{ij}$, for any $i\neq j$;
    \item[2.] $\widetilde{C}_{ii}\le C_{ii}-\alpha$, for any $i$; and
    \item[3.] the Cholesky algorithm implemented in floating-point arithmetic can conclude that 
     $\widetilde{C}$ is positive semi-definite,
  \end{itemize}
  where $\mathbb{F}$ is a floating-point format,
  $\alpha=\frac{(D+1)\kappa}{1-(2D+2)\kappa}\mathtt{tr}(C)+4(D+1)(2(D+2)+\max_i\{C_{ii}\} ) \eta$,
  in which $\kappa$ is the unit roundoff of $\mathbb{F}$ and $\eta$ is the underflow unit of $\mathbb{F}$.
\end{theorem}

\oomit{
Then formula (\ref{sos}) will be translated into positive semi-definite constraints, this is
done by translate polynomial constraints to semi-definite constraints with linear equalities,
a simple way is showing as following.
Let }

\begin{corollary} \label{coro:a}
  Let $\widetilde{C} \in \mathbb{F}^{D\times D}$. Suppose that
  $\frac{(D+1)D \kappa}{1-(2D+2)\kappa}+4(D+1)\eta\le\frac{1}{2}$,
  $\beta = \frac{(D+1)\kappa}{1-(2D+2)\kappa}\mathtt{tr}(\widetilde{C})+4(D+1)(2(D+2)+\max_i\{\widetilde{C}_{ii}\}) \eta>0$,
   where $\mathbb{F}$ is a floating-point format.
  Then $\widetilde{C}+2\beta I\succeq 0$ if the Cholesky algorithm based on floating-point arithmetic succeeds on $\widetilde{C}$,
    i.e., concludes that $\widetilde{C}$ is positive semi-definite.
\end{corollary}
\begin{proof}[of Corollary~\ref{coro:a}]
By directly checking.
 \oomit{
   By Theorem \ref{the:sdp:cond}, treating
    $\widetilde{C}+2\beta I$ as $C$, we just need to prove that
  \begin{align}
  \end{align}

    \label{formula:needp}
    2\beta \ge \alpha(\widetilde{C}+2\beta I),
  \end{align}
  $\alpha(\widetilde{C}+2\beta I)$ means replace $C$ in $\alpha$ by $\widetilde{C}+2\beta I$.
  Since
  \begin{align*}
    & \alpha(\widetilde{C}+2\beta I) \\
    =& \frac{(D+1)\kappa}{1-(2D+2)\kappa}\mathtt{tr}(\widetilde{C}+2\beta I)+4(D+1)(2(D+2)+\max_i(\widetilde{C}+2\beta I)_{ii}) \eta \\
    =&\frac{(D+1)\kappa}{1-(2D+2)\kappa}(\mathtt{tr}(\widetilde{C})+2D\beta )+4(D+1)(2(D+2)+\max_i\widetilde{C}_{ii}+2\beta) \eta \\
    =&\beta + \frac{(D+1)\kappa}{1-(2D+2)\kappa}2D\beta +4(D+1)2\beta \eta\\
    =&\beta +2(\frac{(D+1)D \kappa}{1-(2D+2)\kappa}+4(D+1)\eta)\beta \\
    \le & \beta +2\cdot \frac{1}{2} \beta=2\beta,
  \end{align*}
  thus, (\ref{formula:needp}) holds. }
\end{proof}

According to Remark 5 in \cite{RVS16}, for IEEE 754 binary64 format with rounding to nearest,
$\kappa=2^{-53}(\simeq 10^{-16})$ and $\eta=2^{-1075}(\simeq 10^{-323})$. In this case, the order of magnitude of $\beta$
is $10^{-10}$ and $\frac{(D+1)D \kappa}{1-(2D+2)\kappa}+4(D+1)\eta$ is $10^{-13}$,
 much less than $\frac{1}{2}$. 
Obviously, $\beta$ becomes smaller when the length of binary format becomes  longer.

Without loss of generality, we suppose that the Cholesky algorithm succeed on
$\widehat{C}$ the solution of (\ref{sdp}), which is reasonable since if an SDP solver returns a solution $\widehat{C}$, then $\widehat{C}$ should be considered to be positive semi-definite in a perspective of numeric computation (in other words, we assume the answer obtained by numeric computation is correct.).


Therefore, by Corollary \ref{coro:a}, we have $\widehat{C}+2\beta I\succeq 0$ holds, where $I$ is the identity matrix with corresponding dimension. Then we have
\begin{align*}
  \mathtt{diag}(\widehat{C}_{u_1}, \ldots, \widehat{C}_{u_m}, \widehat{C}_{v_1}, \ldots, \widehat{C}_{v_n}, \widehat{C}_{h-1-uf}, \widehat{C}_{-h-1-vg}) \\
   +2\beta I\succeq 0,
\end{align*}
i.e.,
\begin{equation}
\label{solusdp}
\begin{split}
  &\widehat{C}_{u_1}+2\beta I\succeq 0,~ \ldots,~ \widehat{C}_{u_m}+2\beta I\succeq 0,\\
  &\widehat{C}_{v_1}+2\beta I\succeq 0,~ \ldots,~ \widehat{C}_{v_n}+2\beta I\succeq 0,\\
  &\widehat{C}_{h-1-uf}+2\beta I\succeq 0,~ \widehat{C}_{-h-1-vg}+2\beta I\succeq 0.
\end{split}
\end{equation}

Let $\epsilon=\max_{p\in P,1\le i\le K_p} |\mathtt{tr}(A_{p,i}\widehat{C}_p)-b_{p,i}|$, where  $P=\{u_1,\ldots,u_m,v_1,\ldots,v_n,h\}$, which can be
regarded  as the tolerance  of the SDP solver. Since $|\mathtt{tr}(A_{p,i}C_p)-b_{p,i}|$ is the error term
for each monomial of $p$, i.e., $\epsilon$ can be considered as the error bound on the coefficients of
polynomials $u_i$, $v_j$ and $h$, for any polynomial $\hat{u_i}$ ( $\hat{v_j}$ and $\hat{h}$), computed from (\ref{eqsforu})
by replacing $C_u$ with the corresponding $\widehat{C_u}$,
there exists a corresponding remainder
term $R_{u_i}$ (resp. $R_{v_j}$ and $R_h$) with degree not greater than
$2d$, whose coefficients are bounded by $\epsilon$.
Hence, from (\ref{solusdp}), we have
\begin{eqnarray} \label{solusos}
\begin{split}
  \widehat{u_i}+R_{u_i}+2\beta E_d(\xx,\yy)^T E_d(\xx,\yy)\in\sum\RR[\xx,\yy]^2, \\ i=1,\ldots,m, \\
  \widehat{v_j}+R_{v_j}+2\beta E_d(\xx,\zz)^T E_d(\xx,\zz) \in\sum\RR[\xx,\zz]^2,  \\
  j=1,\ldots,n,  \\
  \widehat{h}+R_{h}-1-\sum_{i=1}^m (\widehat{u_i}+R'_{u_i}) f_i+2\beta E_d(\xx,\yy)^T E_d(\xx,\yy)
   \in\sum\RR[\xx,\yy]^2, \\
  -\widehat{h}+R'_{h}-1-\sum_{j=1}^m (\widehat{v_j}+R'_{v_j}) g_j+2\beta E_d(\xx,\zz)^T E_d(\xx,\zz) \in\sum\RR[\xx,\zz]^2.
\end{split}
\end{eqnarray}
\oomit{where $R_{u_i}$, $R_{v_j}$ or $R_h$ is not a polynomial, but a numerical error term, which mean for each
SOS constraint, there exists a polynomial with coefficients bounded by $\epsilon$ and degree bounded
by $2d$ as a numerical error correction term to make the constraint holds.}

Now, in order to avoid unsoundness of our approach caused by the numerical issue due to SDP, we have to prove
\begin{align}
  f_1\ge 0\wedge\cdots\wedge f_m\ge0 \Rightarrow \widehat{h}>0,\label{numhold:1}\\
  g_1\ge 0\wedge\cdots\wedge g_n\ge0 \Rightarrow \widehat{h}<0.\label{numhold:2}
\end{align}

Regarding (\ref{numhold:1}), let $R_{2d,\xx}$ be a polynomial in $\RR[|\xx|]$, whose total degree is
$2d$, and all coefficients are $1$, e.g., $R_{2,x,y}=1+|x|+|y|+|x^2|+|xy|+|y^2|$.
Since $S=\{(\xx,\yy) \mid f_1\ge 0\wedge\cdots\wedge f_m\ge0\}$ is a compact set, then for any polynomial
$p\in\RR[\xx,\yy]$, $|p|$ is bounded on $S$. Let $M_1$ be an upper bound of $R_{2d,\xx,\yy}$ on $S$, $M_2$
an upper bound of $E_d(\xx,\yy)^T E_d(\xx,\yy)$, and $M_{f_i}$ an upper bound of $f_i$ on $S$.
Then, $|R_{u_i}|$, $|R'_{u_i}|$ and $|R_{h}|$ are bounded by $\epsilon M_1$.
Let $E_{\xx \yy}=E_d(\xx,\yy)^T E_d(\xx,\yy)$.
So for any $(\xx_0,\yy_0)\in S$, considering
  the polynomials below at $(\xx_0,\yy_0)\in S$, by the first and third line in  (\ref{solusos}),  we have
\begin{align*}
  \widehat{h}\ge & 1-R_{h}+\sum_{i=1}^m (\widehat{u_i}+R'_{u_i}) f_i-2\beta E_{\xx \yy}\\
  \ge& 1- \epsilon M_1+\sum_{i=1}^m(\widehat{u_i}+R_{u_i}+2\beta E_{\xx \yy}+R'_{u_i}-R_{u_i}-2\beta E_{xy}) f_i-2\beta M_2 \\
  =& 1- \epsilon M_1-2\beta M_2 +\sum_{i=1}^m(\widehat{u_i}+R_{u_i}+2\beta E_{\xx \yy}) f_i+\sum_{i=1}^m(R'_{u_i}-R_{u_i}-2\beta E_{\xx \yy}) f_i \\
  \ge &1- \epsilon M_1-2\beta M_2 +0-\sum_{i=1}^m (\epsilon M_1+\epsilon M_1+2\beta M_2)M_{f_i} \\
  =& 1-(2\sum_{i=1}^m M_{f_i}+1)M_1 \epsilon- 2(\sum_{i=1}^m M_{f_i}+1)M_2 \beta.
\end{align*}
Whence,
\begin{equation}
\label{forsolhold:1}
\begin{split}
    &f_1\ge 0\wedge\cdots\wedge f_m\ge0 \Rightarrow \\
  &\widehat{h}\geq 1-(2\sum_{i=1}^m M_{f_i}+1)M_1 \epsilon- 2(\sum_{i=1}^m M_{f_i}+1)M_2 \beta.
\end{split}
\end{equation}

Let $S'=\{(\xx,\zz) \mid g_1\ge 0\wedge\cdots\wedge g_n\ge0\}$, $M_3$ be an upper bound of $R_{2d,\xx,\zz}$ on $S'$, $M_4$
an upper bound of $E_d(\xx,\zz)^T E_d(\xx,\zz)$ on $S'$, and $M_{g_j}$ an upper bound of $g_j$ on $S'$.
Similarly to the above, it follows
\begin{align*}\label{forsolhold:2}
  &g_1\ge 0\wedge\cdots\wedge g_n\ge0 \Rightarrow \\
  &-\widehat{h}\geq 1-(2\sum_{j=1}^n M_{g_j}+1)M_3 \epsilon- 2(\sum_{j=1}^n M_{g_j}+1)M_4 \beta.
\end{align*}

So,  the following proposition holds.
\begin{proposition} \label{prop:17}
  There exist two positive constants $\gamma_1$ and $\gamma_2$ such that
  \begin{align}
  f_1\ge 0\wedge\cdots\wedge f_m\ge0 \Rightarrow \widehat{h}\geq 1-\gamma_1 \epsilon- \gamma_2 \beta, \\
  g_1\ge 0\wedge\cdots\wedge g_n\ge0 \Rightarrow -\widehat{h}\geq 1-\gamma_1 \epsilon- \gamma_2 \beta.
\end{align}
\end{proposition}
\begin{proof}[of Proposition~\ref{prop:17}]
  We just need to take
\begin{align*}
 \gamma_1=\max((2\sum_{i=1}^m M_{f_i}+1)M_1,(2\sum_{j=1}^n M_{g_j}+1)M_3)\, ,  \\
 \gamma_2=\max(2(\sum_{i=1}^m M_{f_i}+1)M_2,2(\sum_{j=1}^n M_{g_j}+1)M_4)
\end{align*}
 in formulas (\ref{forsolhold:1}) and (\ref{forsolhold:2}).
\end{proof}

Since $\epsilon$ and $\beta$ heavily rely on the numerical tolerance and the floating point representation, it is easy to see that
 $\epsilon$ and $\beta$ become small enough with
  $\gamma_1 \epsilon<\frac{1}{2}$ and $\gamma_2 \beta<\frac{1}{2}$,
 if  the numerical tolerance is small enough and the length of the floating point representation is  long enough. This implies that
\begin{align*}
  f_1\ge 0\wedge\cdots\wedge f_m\ge0 \Rightarrow \widehat{h}>0, \\
  g_1\ge 0\wedge\cdots\wedge g_n\ge0 \Rightarrow -\widehat{h}>0.
\end{align*}
If so, any numerical result $\widehat{h}>0$ returned by calling
 an SDP solver to (\ref{sdp}) is guaranteed
to be a real interpolant for $\phi$ and $\psi$, i.e., a correct solution to \textbf{Problem} 1.

\begin{example}
  Consider the numerical result for Example \ref{ex3} in Section \ref{sec:ssc}. Let
  $M_{f_1}$, $M_{f_2}$, $M_{f_3}$, $M_{g_1}$, $M_{g_2}$, $M_{g_3}$, $M_1$, $M_2$, $M_3$,
  $M_4$ are defined as above. It is easy to see that
  \begin{align*}
    f_1\ge 0\Rightarrow & |x|\le 2\wedge |y|\le 2\wedge|z|\le 2\wedge|a_1|\le 2\wedge|b_1|\le 2 \wedge|c_1|\le 2\wedge|d_1|\le 2.
  \end{align*}
  Then, by simple calculations, we obtain
  $$M_{f_1}=4, M_{f_2}=32, M_{f_3}=3, M_1=83, M_2=29.$$
  Thus, $$(2\sum_{i=1}^m M_{f_i}+1)M_1=6557,\quad 2(\sum_{i=1}^m M_{f_i}+1)M_2=2320.$$
Also, since
  \begin{align*}
    g_1\ge 0\Rightarrow & |x|\le 2\wedge |y|\le 2\wedge|z|\le 2\wedge|a_2|\le 2\wedge
      |b_2|\le 2  \wedge|c_2|\le 2\wedge|d_2|\le 2,
  \end{align*}
  we obtain
  $$M_{g_1}=4, M_{g_2}=7, M_{g_3}=2, M_3=83, M_4=29.$$
  Thus, $$(2\sum_{i=1}^m M_{g_i}+1)M_3=2241,\quad 2(\sum_{i=1}^m M_{g_i}+1)M_4=812.$$
Consequently, we have $\gamma_1=6557$ and $\gamma_2=2320$ in Proposition \ref{prop:17}.

Due to the fact that the default error tolerance is $10^{-8}$ in the SDP solver Mosek and $h$ is rounding to $4$ decimal places, we have $\epsilon=\frac{10^{-4}}{2}$. In addition,
 as the absolute value of each element in $\widehat{C}$ is less than $10^3$, and the dimension of $D$ is less than $10^3$, we obtain that
  \begin{align*}
    \beta  &= \frac{(D+1)\kappa}{1-(2D+2)\kappa}\mathtt{tr}(\widetilde{C})
  +4(D+1)(2(D+2)+\max_i(\widetilde{C}_{ii})) \eta \\
   & \le \frac{(1000+1)10^{-16}}{1-(2000+2)10^{-16}}10^6
  +4(1000+1)(2(1000+2)+1000) 10^{-1075} \\
  &\le 10^{-6}.
  \end{align*}

Consequently,
  \begin{align*}
    \gamma_1 \epsilon \le 6557 \cdot \frac{10^{-4}}{2} <\frac{1}{2},\\
    \gamma_2 \beta\le 2320 \cdot 10^{-6} < \frac{1}{2},
  \end{align*}
which imply that $h(x,y,z)>0$ presented in Example \ref{ex3} is indeed a sound interpolant.

  \begin{remark} \label{remark:1}
  Besides, the result could be verified by the following symbolic computation procedure instead: computing $P_\xx(\phi)$ and $P_\xx(\psi)$ first by some symbolic tools, such as Redlog \cite{dolzmann1997} which is a package that extends the computer algebra system REDUCE to a computer logic system; then verifying $\xx\in P_\xx(\phi)\Rightarrow h(\xx)>0$ and $\xx\in P_\xx(\psi)\Rightarrow h(\xx)<0.$ For this example, $P_{x,y,z}(\phi)$ and $P_{x,y,z}(\psi)$ obtained by Redlog are too complicated and therefore not presented here. The symbolic computation can verify that $h(x,y,z)$ in this example is exactly an interpolant, which confirms  our conclusion.
  \end{remark}
\end{example}

\section{Generalizing to general polynomial formulas} \label{sec:generalization}

\begin{problem} \label{problem:2}
Let $\phi(\xx,\yy)$ and $\psi(\xx,\zz)$ be two polynomial formulas defined as follows,
\begin{align*}
  \phi(\xx,\yy) : \bigvee_{i=1}^m \phi_i, \quad \phi_i=\bigwedge_{k=1}^{K_i}f_{i,k}(\xx,\yy) \geq 0 ; \\
  \psi(\xx,\zz) : \bigvee_{j=1}^n \psi_j, \quad \psi_j=\bigwedge_{s=1}^{S_j}g_{j,s}(\xx,\zz) \geq 0 ,
\end{align*}
where all $f_{i,k}$ and $g_{j,s}$ are polynomials.
Suppose $\phi \wedge \psi \models \bot$, and for $i=1,\ldots,m$, $j=1,\ldots,n$, $\{(\xx,\yy)\mid \phi_i(\xx,\yy)\}$ and $\{(\xx,\zz)\mid \psi_j(\xx,\zz)\}$ are all semi-algebraic sets of the Archimedean form. Find a polynomial $h(\xx)$ such that $h(\xx)>0$ is an interpolant for $\phi$ and $\psi$.
\end{problem}

\begin{theorem}
\label{the:main:forgen}
For \textbf{Problem} 2,
there exists a polynomial $h(\xx)$ satisfying
 \[\forall \xx \in P_{\xx}(\phi(\xx,\yy)), h(\xx) > 0,\]
 \[\forall \xx \in P_{\xx}(\psi(\xx,\zz)), h(\xx) < 0.\]
\end{theorem}

\begin{proof}[of Theorem~\ref{the:main:forgen}]
  We claim that Lemma \ref{the:forproj} holds for \textbf{Problem} 2 as well. Since $\{(\xx,\yy)\mid \phi_i(\xx,\yy)\}$ and $\{(\xx,\zz)\mid \psi_j(\xx,\zz)\}$ are all semi-algebraic sets of the Archimedean form, then $\{(\xx,\yy)\mid \phi(\xx,\yy)\}$ and $\{(\xx,\zz)\mid \psi(\xx,\zz)\}$ both are compact. See $\{(\xx,\yy)\mid \phi(\xx,\yy)\}$ or
  $\{(\xx,\zz)\mid \psi(\xx,\zz)\}$ as $S$ in the proof of Lemma \ref{the:forproj}, then
  Lemma \ref{the:forproj} holds for \textbf{Problem} 2.
 Thus, the rest of proof is same as that for Theorem \ref{the:main}.
\end{proof}


\begin{corollary} \label{conse:forgen}
  Let $\phi(\xx,\yy)$ and $\psi(\xx,\zz)$ be defined as in \textbf{Problem} 2. There must exist a polynomial $h(\xx)$ such that $h(\xx)>0$ is an interpolant for $\phi$ and $\psi$.
\end{corollary}


\begin{theorem} \label{the:gencase:main}
  Let $\phi(\xx,\yy)$ and $\psi(\xx,\zz)$ be defined as in \textbf{Problem} 2. Then there exists a polynomial $h(\xx)$ and $\sum_{i=1}^m(K_i+1)+\sum_{j=1}^n (S_j+1)$ sum of squares polynomials $u_{i,k}(\xx,\yy)$ $(i=1,\ldots,m$, $k=1,\ldots,K_i+1)$, $v_{j,s}(\xx,\zz)$
   $(j=1,\ldots,n$, $s=1,\ldots,S_j)$ satisfying the following semi-definite constraints such that $h(\xx)>0$ is an interpolant for $\phi(\xx,\yy)$ and $\psi(\xx,\zz)$:
  \begin{align}
  \label{formula:main:casegen}
   &h-1=\sum_{k=1}^{K_i} u_{i,k}f_{i,k}+u_{i,K_i+1},\quad i=1,\ldots, m;\\
   &-h-1=\sum_{s=1}^{S_j} v_{j,s}g_{j,s}+ v_{j,S_j+1}, \quad j=1,\ldots, n.
  \end{align}
\end{theorem}

\begin{proof}[of Theorem~\ref{the:gencase:main}]
  According to  the property of Archimedean, the proof is  same as that for Theorem \ref{the:main2}.
\end{proof}

Similarly, \textbf{Problem} 2 can be equivalently reformulated as the problem of searching for sum of squares polynomials satisfying 
 \begin{equation}
\label{sos:casegen} \left\{
\begin{split}
    &h(\xx)-1-\sum_{k=1}^{K_i} u_{i,k}f_{i,k} \in \sum \mathbb{R}[\xx,\yy]^2, i=1,\ldots,m;\\
   &-h(\xx)-1-\sum_{s=1}^{S_j} v_{j,s}g_{j,s}\in \sum \mathbb{R}[\xx,\zz]^2, j=1,\ldots,n;\\
   &u_{i,k} \in \sum\mathbb{R}[\xx,\yy]^2, i=1,\ldots,m, k=1,\ldots, K_i;\\
   &v_{j,s} \in \sum\mathbb{R}[\xx,\zz]^2, j=1,\ldots,n, s=1,\ldots, S_j.
\end{split} \right.
\end{equation}

\begin{example} \label{ex:317}
  Consider 
{\small  \begin{align*}
    \phi(x,y,a_1,a_2,b_1,b_2): (f_1\ge0\wedge f_2\ge0)\vee(f_3\ge0 \wedge f_4\ge0),\\
    \psi(x,y,c_1,c_2,d_1,d_2): (g_1\ge0\wedge g_2\ge0)\vee(g_3\ge0 \wedge g_4\ge0),
  \end{align*}
}
  where
{
  \begin{align*}
    f_1=&16-(x+y-4)^2-16(x-y)^2-a_1^2,\\
    f_2=&x+y-a_2^2-(2-a_2)^2,\\
    f_3=&16-(x+y+4)^2-16(x-y)^2-b_1^2,\\
    f_4=&-x-y-b_2^2-(2-b_2)^2,\\
    g_1=&16-16(x+y)^2-(x-y+4)^2-c_1^2,\\
    g_2=&y-x-c_2^2-(1-c_2)^2,\\
    g_3=&16-16(x+y)^2-(x-y-4)^2-d_1^2,\\
    g_4=&x-y-d_2^2-(1-d_2)^2.
  \end{align*}
}
We get a concrete $\sdp$ problem of the form (\ref{sos:casegen}) by setting the degree of $h(x,y)$ in (\ref{sos:casegen}) to be $2$. Using the MATLAB package YALMIP and Mosek, we obtain
\[    h(x,y)=-2.3238+0.6957x^2+0.6957y^2+7.6524xy.\]
The result is plotted in Fig. \ref{fig317}, and can be verified either by
numerical error analysis as in Example \ref{ex3} or by a symbolic procedure like REDUCE as described in Remark \ref{remark:1}.

\begin{figure}
\flushleft
\includegraphics[scale=0.5]{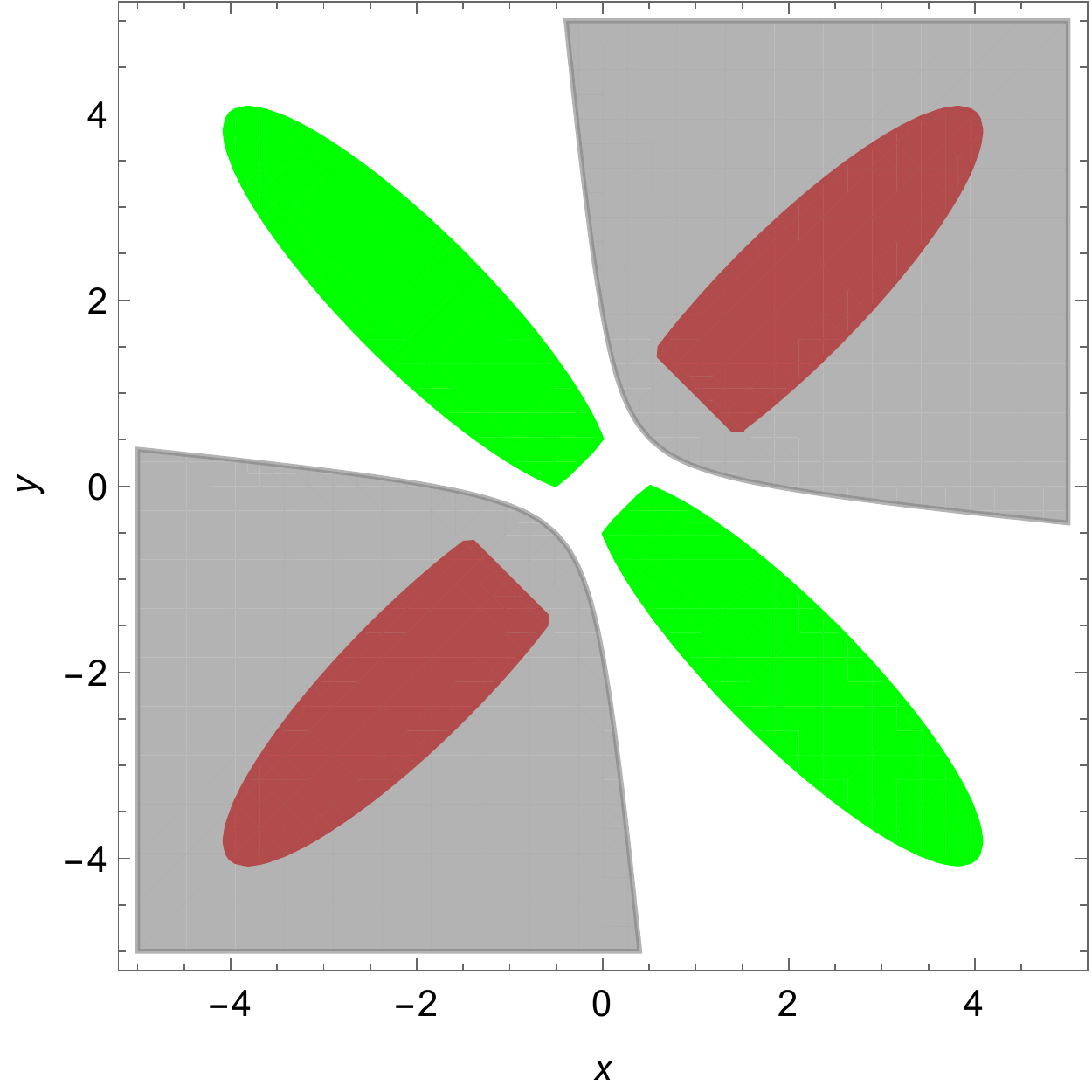}
\caption{Example \ref{ex:317}. \small{(Red region: ~ $P_{x,y}(\phi(x,y,a_1,a_2,b_1,b_2))$; Green region:~ $P_{x,y}(\psi(x,y,c_1,c_2,d_1,d_2))$; Gray region:~  $\{(x,y)\mid h(x,y)>0\}$.)}}
\label{fig317}
\end{figure}
\end{example}

\begin{example}[Ultimate] \label{ex:ultimate}
  Let
  \begin{eqnarray*}
\phi = (f_1\ge 0 \wedge f_2\ge 0 \vee f_3\ge 0) \wedge f_4\ge 0 \wedge f_5\ge 0 \vee f_6\ge 0,
     \\[1mm]
\psi = (g_1\ge 0 \wedge g_2\ge 0 \vee g_3\ge 0) \wedge g_4\ge 0 \wedge g_5\ge 0 \vee g_6\ge 0,
\end{eqnarray*}
where
\begin{eqnarray*}
 f_1& = &3.8025-x^2-y^2, \\
 f_2 & = &y, \\
 f_3 & = & 0.9025-(x-1)^2-y^2, \\
 f_4& = & (x-1)^2+y^2-0.09, \\
 f_5 & = & (x+1)^2+y^2-1.1025, \\
 f_6 & = & 0.04-(x+1)^2-y^2,\\
 g_1 & = & 3.8025-x^2-y^2, \\
 g_2 & = & -y,  \\
 g_3 & = & 0.9025-(x+1)^2-y^2, \\
 g_4 & = & (x+1)^2+y^2-0.09, \\
 g_5 & = & (x-1)^2+y^2-1.1025,\\
 g_6 &= & 0.04-(x-1)^2-y^2.
\end{eqnarray*}
We first convert $\phi$ and $\psi$ to the disjunction normal form as:
\begin{align*}
\phi =&(f_1\ge 0 \wedge f_2\ge 0 \wedge f_4\ge 0 \wedge f_5 \ge 0) \\
&\vee (f_3\ge 0\wedge f_4\ge 0 \wedge f_5\ge 0) \vee (f_6\ge 0),
     \\[1mm]
\psi = & (g_1\ge 0 \wedge g_2\ge 0 \wedge g_4\ge 0 \wedge g_5\ge 0) \\
&\vee (g_3\ge 0 \wedge g_4\ge 0 \wedge g_5\ge 0) \vee (g_6\ge 0).
\end{align*}
We get a concrete $\sdp$ problem of the form (\ref{sos:casegen}) by setting the degree of $h(x,y)$ in (\ref{sos:casegen}) to be $7$. Using the MATLAB package YALMIP and Mosek, keeping the decimal to four, we obtain
{  \begin{align*}
    h&(x,y)=1297.5980x+191.3260y-3172.9653x^3+196.5763x^2y\\
    &+2168.1739xy^2+1045.7373y^3
    +1885.8986x^5-1009.6275x^4y\\
    &+3205.3793x^3y^2-1403.5431x^2y^3+
    1842.0669xy^4\\
    &+1075.2003y^5-222.0698x^7+
    547.9542x^6y-704.7474x^5y^2\\
    &+1724.7008x^4y^3-728.2229x^3y^4+1775.7548x^2y^5\\
    &-413.3771xy^6+1210.2617y^7.
  \end{align*}
}
The result is plotted in Fig. \ref{ultimate}, and can be verified either by
numerical error analysis as in Example \ref{ex3} or by a symbolic procedure like REDUCE as described in Remark \ref{remark:1}.

\begin{figure}
\flushleft
\includegraphics[scale=0.35]{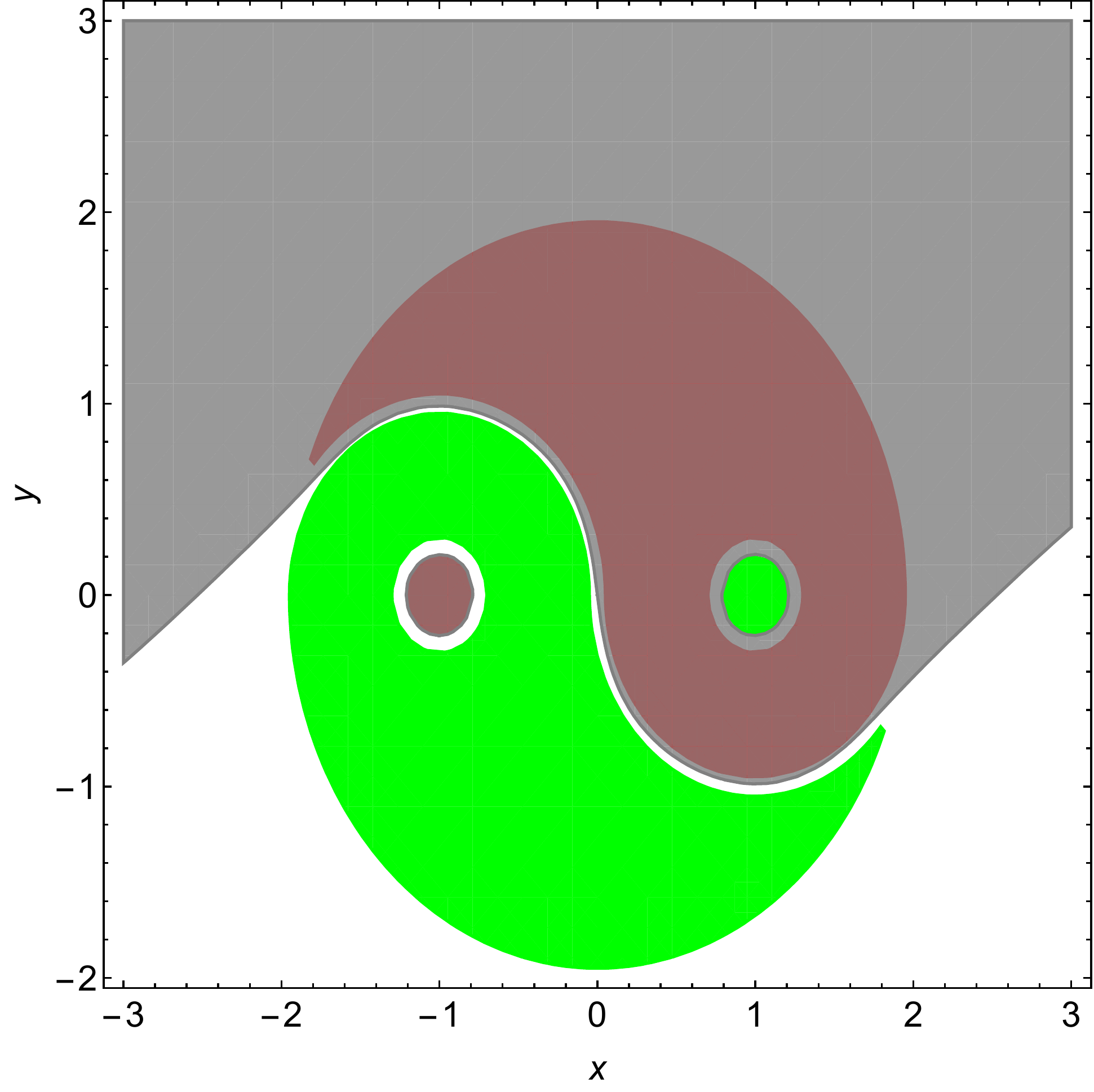}
\caption{Example \ref{ex:ultimate}.  \small{(Red region: $P_{x,y}(\phi(x,y))$;  Green region: $P_{x,y}(\psi(x,y))$; Gray region: $\{(x,y)\mid h(x,y)>0\}$.)}}
\label{ultimate}
\end{figure}
\end{example}

\section{Application to Invariant Generation} \label{sec:invariant}
In this section, as an application, we show how to
 apply our approach to invariant generation in program verification.

In \cite{LSXLSH2017}, Lin \emph{et al.} proposed a framework for invariant generation using \emph{weakest precondition},
\emph{strongest postcondition} and \emph{interpolation}, which consists of two procedures, i.e.,
synthesizing invariants by forward interpolation based on \emph{strongest postcondition} and
\emph{interpolant generation},  and by backward interpolation  based on \emph{weakest precondition} and
\emph{interpolant generation}. In \cite{LSXLSH2017}, only linear invariants can be synthesized as no powerful
approaches are available to synthesize nonlinear interpolants. Obviously, our results can
strengthen their framework by allowing to generate nonlinear invariants.
To this end, we revise the two procedures, i.e., Squeezing Invariant - Forward
 and Squeezing Invariant - Backward, in their
framework accordingly, and obtain
 Algorithm \ref{forward:invariant} and  Algorithm \ref{backward:invariant}.
  The major revisions include:
\begin{itemize}
	\item   firstly, we exploit our method to synthesize interpolants
	see line 4 in Algorithm \ref{forward:invariant} and line 4 in in Algorithm \ref{backward:invariant};
	\item secondly, we add a conditional statement for $A_{i+1}$ at line 7-10 in Algorithm \ref{forward:invariant} in order to make $A_{i+1}$ to be Archimedean, the same for $B_{j+1}$ in Algorithm \ref{backward:invariant}.
\end{itemize}
We then illustrate the basic idea by exploiting Algorithm \ref{forward:invariant}
to an example given in Algorithm~\ref{exam:inv:2}. 
The reader can refer to \cite{LSXLSH2017} for the detail of the framework.

\begin{algorithm}
\caption{Revised Squeezing Invariant - Forward}
\label{forward:invariant}
	\begin{algorithmic}[1]
	\REQUIRE{An annotated loop: $\{P\}$ while $\rho$ do $C$ $\{Q \}$}
	\ENSURE{ (yes/no, $\II$), where $\II$ is a loop invariant}
	
	\STATE{$A_0 \gets P$; $B_0 \gets (\neg \rho \wedge\neg Q)$; $i \gets 0$; $j\gets 0$;}
	\WHILE{$\top$}
      \IF{$(\bigvee_{k=0}^i A_i)\wedge B_j$ is not satisfiable}
        \STATE{call our method to synthesize an interpolant for $(\bigvee_{k=0}^i A_i)$ and $B_j$, say  $\II_i$;}

        \COMMENT{{\rm Use our method to generate interpolant}}
        \IF{$\{\II_i \wedge \rho\} \, C \,  \{ \II_i\}$}
          \RETURN{(yes, $\II_i$);}
        \ELSIF{$\II_i$ is bounded}
          \STATE{$A_{i+1}\gets\mathsf{sp}(\II_i\wedge\rho,C)$;}
        \ELSE
          \STATE{$A_{i+1} \gets \mathsf{sp}(A_i\wedge\rho,C)$;}
        \ENDIF

        \COMMENT{$\mathsf{sp}$: a predicate transformer to compute the strongest postcondition of $C$ w.r.t. $\II_i\wedge \rho$}
        \STATE{$i \gets i+1$;}
        \STATE{$B_{j+1}\gets B_0\vee (\rho\wedge \mathsf{wp}(C,B_j))$;}

        \COMMENT{$\mathsf{wp}$: a predicate transformer to compute the weakest precondition of $C$ w.r.t. $B_j$}
        \STATE{$i\gets i+1$;}
      \ELSIF{$A_i$ is concrete}
        \RETURN{(no, $\bot$);}
      \ELSE
        \WHILE{$A_i$ is not concrete}
          \STATE{$i \gets i-1$;}
        \ENDWHILE
        \STATE{$A_{i+1}\gets \mathsf{sp}(A_i\wedge\rho,C)$;}
        \STATE{$i\gets i+1$;}
      \ENDIF
    \ENDWHILE
\end{algorithmic}
\end{algorithm}

\begin{algorithm}
\caption{Revised Squeezing Invariant - Backward}
\label{backward:invariant}
	\begin{algorithmic}[1]
	\REQUIRE{An annotated loop: $\{P\}$ while $\rho$ do $C$ $\{Q \}$}
	\ENSURE{ (yes/no, $\II$), where $\II$ is a loop invariant}
	
	\STATE{$A_0 \gets P$; $B_0 \gets (\neg \rho \wedge\neg Q)$; $i \gets 0$; $j\gets 0$;}
	\WHILE{$\top$}
      \IF{$B_j\wedge(\bigvee_{k=0}^i A_i)$ is not satisfiable}
        \STATE{call our method to synthesize an interpolant for $B_j\wedge(\bigvee_{k=0}^i A_i)$, say $I_j$;}

        \COMMENT{{\rm Use our method to generate interpolant}}
        \IF{$\{\neg\II_j \wedge \rho\}C \{ \neg\II_j\}$}
          \RETURN{(yes, $\neg\II_j$);}
        \ELSIF{$\II_j$ is bounded}
          \STATE{$B_{j+1} \gets \II_j \vee(\rho\wedge \mathsf{wp}(C,\II_j))$;}
        \ELSE
          \STATE{$B_{j+1} \gets B_j \vee(\rho\wedge \mathsf{wp}(C,B_j))$;}
        \ENDIF
        \STATE{$j \gets j+1$;}
        \STATE{$A_{i+1} \gets \mathsf{sp}(A_i\wedge\rho,C)$;}
        \STATE{$i\gets i+1$;}
      \ELSIF{$B_j$ is concrete}
        \RETURN{(no, $\bot$);}
      \ELSE
        \WHILE{$B_j$ is not concrete}
          \STATE{$j \gets j-1$;}
        \ENDWHILE
        \STATE{$B_{j+1}\gets B_0\vee (\rho \wedge \mathsf{wp}(C,B_j))$;}
        \STATE{$j\gets j+1$;}
      \ENDIF
    \ENDWHILE
\end{algorithmic}
\end{algorithm}

\begin{example}
	Consider a while loop given in Algorithm \ref{exam:inv:2}, which is adapted  from \cite{DDLM13}
	by modifying the precondition and the postcondition so that the precondition and the negation of the postcondition are nonlinear and compact.
	\begin{algorithm}
		\caption{} \label{exam:inv:2}
		\begin{algorithmic}[1]
			\STATE /* Pre: $100-(x+50)^2\ge0\wedge100-y^2\ge0$ */
			\WHILE {$x<0$}
			\STATE $x\gets x+y$;
			\STATE $y\gets y+1$;
			\ENDWHILE
			\STATE /* Post: $(x-5)^2+(y+3)^4> 100$ */
		\end{algorithmic}
	\end{algorithm}
	We apply  the algorithm Squeezing Invariant - Forward in Algorithm \ref{forward:invariant} to the loop to
	compute an invariant which can witness its correctness.
	
	Firstly, at line 1 in Algorithm \ref{forward:invariant}, we have $\rho:x<0$ and
	\begin{align*}
		A_0:&~ 100-(x+50)^2\ge0\wedge100-y^2\ge0, \\
		B_0:&~ x\ge0 \wedge 100-(x-5)^2-(y+3)^4\ge 0.
	\end{align*}
	Then, at line 3, $A_0\wedge B_0 \models \bot$. Using our method, we can synthesize an interpolant for
	$A_0$ and $B_0$ (line 4) as:
	\begin{align*}
		\II_0: & -4.8031-6.8601x+4.5900y+0.0905x^2
		-0.5331y^2 \\
	 & 	+0.1376xy> 0.
	\end{align*}
	It can be checked that $\{\II_0\wedge \rho \} \ C \ \{\II_0\}$ does not hold ( line 5), where $C$ stands for the loop body.
	
	Secondly,  since $\II_0$ is not bounded, set $A_1=\mathsf{sp}(A_0\wedge \rho,C)$ (line 10), and $B_1=B_0\vee(\rho \wedge \mathsf{wp}(C,B_0))$ (line 12), i.e.,
	\begin{align*}
		A_1: & ~x=x'+y' \wedge y=y'+1 \wedge x'<0
		     \wedge
		100-(x'+50)^2\ge0 \wedge 100-y'^2\ge0,\\
		B_1: & ~B_0 \vee (x<0 \wedge x''=x+y \wedge y''=y+1
		   \wedge
		x''\ge0 \wedge 100-x''^2-(y''+3)^4\ge 0).
	\end{align*}
	Now, repeating  the while loop once again, at line 3, we have $(A_0\vee A_1)\wedge B_1$ is not
	satisfiable. Thus, with our method, we can obtain
	\begin{align*}
		\II_1: & -5.5937-10.6412x+7.9251y +0.1345x^2
		+0.3086y^2
		  +0.0020xy> 0.
	\end{align*}
	It can be checked that $\{\II_1\wedge \rho\}\ C \, \{\II_1\}$ holds.
	Thus, the algorithm will return $({\rm yes}, \II_1)$.
	
	Since $\II_1$ is an interpolant of $(A_0\vee A_1)\wedge B_1$,
	it follows that
	$(A_0\vee A_1)\models \II_1$ and $\II_1\wedge B_1\models\bot$.
	From $(A_0\vee A_1)\models \II_1$,  we have $\mathrm{Pre} \models \II_1$
	as $\mathrm{Pre}=A_0$. Moreover,
	from $\II_1\wedge B_1\models\bot$ and $B_1=B_0\vee(\rho\wedge \mathsf{wp}(C,B_0))$, we have
	$\II_1\wedge B_0\models\bot$, i.e., $\II_1\wedge (\neg \mathrm{Post}  \wedge(\neg\rho))\models\bot$. This implies that $\II_1\wedge (\neg\rho)\models Q$. Hence, we have
	\begin{align*}
		\mathrm{Pre} \models \II_1,~
		\II_1\wedge (\neg\rho)\models \mathrm{Post},~
		\{\II_1\wedge \rho\}\, C\, \{\II_1\},
	\end{align*}
	i.e., $\II_1$ is an inductive invariant that can prove the correctness of the annotated loop in Algorithm
	\ref{exam:inv:2}.
\end{example}

\begin{example}
	Consider another loop given in Algorithm \ref{exam:program1} for controlling
	the acceleration of a car adapted from
	\cite{KB11}. Suppose we know that $\textit{vc}$ is in $[0,40]$ at the beginning of the loop, we would like to prove that $\textit{vc}<49.61$ holds after the loop.
	Since the loop guard is unknown, it means that the loop may terminate after any number of iterations. 
	
	\begin{algorithm}
		\caption{Control code for accelerating a car} \label{exam:program1}
		\begin{algorithmic}[1]
			\STATE /* Pre: $\textit{vc}\in[0,40]$ */
			\WHILE {unknown}
			\STATE $\textit{fa}\gets 0.5418*\textit{vc}*\textit{vc}$;
			\STATE $fr\gets 1000-\textit{fa}$;
			\STATE $\textit{ac}\gets 0.0005*fr$;
			\STATE $\textit{vc}\gets \textit{vc}+\textit{ac}$;
			\ENDWHILE
			\STATE /* Post: $\textit{vc}<49.61$ */
		\end{algorithmic}
	\end{algorithm}
	
	We apply  Algorithm \ref{forward:invariant} to the computation of an invariant to ensure that $\textit{vc}<49.61$ holds. Since $\textit{vc}$ is the velocity of car, $0\le \textit{vc}<49.61$ is required to hold in order to maintain safety. 
	Via Algorithm \ref{forward:invariant}, we have
	$A_0=\{\textit{vc}\mid \textit{vc}(40-\textit{vc})\ge0\}$ and $B=\{\textit{vc}\mid \textit{vc}<0\}\cup\{ \textit{vc}\mid \textit{vc}\ge 49.61\}$.
	Here,  we replace $B$ with $B'=[-2,-1]\cup[49.61,55]$), which is in order to make $B$ to compact,
	i,e, $B'=\{\textit{vc}\mid (\textit{vc}+2)(-1-\textit{vc})\ge 0\vee (\textit{vc}-49.61)(55-\textit{vc})\ge0\}$, in order to make it with Archimedean form.
	
	
	Firstly, it is evident that
	$A_0: \textit{vc}(40-\textit{vc})\ge0$ implies $A_0\wedge B'\models \bot$. By applying our approach,
	we obtain an interpolant
	\begin{align*}
		\II_0: 1.4378+3.3947*\textit{vc}-0.083*\textit{vc}^2>0
	\end{align*} for $A_0$ and $B'$.
	It is verified  that $\{\II_0\} \, C \, \{\II_0\}$  (line 5) does not hold, where $C$ stands for the loop body.
	
	Secondly, by setting $A_1=sp(\II_0,C)$ (line 8) and re-calling our approach, we obtain an interpolant $$\II_1: 2.0673+3.0744*\textit{vc}-0.0734*\textit{vc}^2>0$$ for
	$A_0\cup A_1$ and $B'$.
	Likewise, it is verified that
	$\{\II_1\}\, C \, \{\II_1\}$  (line 5) does not hold.
	
	Thirdly, repeating the above procedure again, we obtain an interpolant
	$$\II_2: 2.2505+2.7267*\textit{vc}-0.063*\textit{vc}^2>0,$$ and it is verified that  $\{\II_2\} \, C \, \{\II_2\}$ holds, implying that $\II_2$ is an invariant.
	Moreover,  it is trivial to verify that $\II_2\Rightarrow \textit{vc}<49.61$.
	
	Consequently, we have the conclusion that $\II_2$ is an inductive invariant which witnesses
	the correctness of the loop.
\end{example}

\section{Conclusion}\label{sec:con}
In this paper we propose a  sound and complete method to synthesize Craig interpolants for mutually contradictory polynomial formulas $\phi(\xx,\yy)$ and $\psi(\xx,\zz)$, with the form $f_1\geq 0\wedge\cdots\wedge f_n\geq 0$, where $f_i$'s are polynomials in $\xx,\yy$ or $\xx,\zz$ and the quadratic module generated by $f_i$'s is Archimedean. The interpolant could be generated by solving a semi-definite programming problem, which is a generalization of the method in \cite{DXZ13} dealing with mutually contradictory formulas with the same set of variables and the method in \cite{GDX16} dealing with mutually contradictory formulas with concave quadratic polynomial inequalities. 
As an application, we apply our approach to invariant generation in program verification.


As a future work, 
we would like to consider interpolant synthesizing for formulas with strict polynomial inequalities. 
Also, it deserves to consider how to synthesize interpolants for the combination of non-linear formulas and
other theories based on our approach and other existing ones.

\bibliographystyle{spmpsci}

\bibliography{references}
\end{document}